\documentclass[runningheads]{llncs}

\title{Envy-free matchings with cost-controlled quotas}

\author{Girija Limaye \and 
Meghana Nasre} 

\institute{Indian Institute of Technology Madras, India\\
\email{\{girija,meghana\}@cse.iitm.ac.in}}

\authorrunning{G. Limaye, M. Nasre}

\usepackage{fullpage} 

\usepackage{array}
\usepackage{enumitem}
\usepackage{amsmath}
\usepackage{graphicx}
\usepackage{algorithm}
\usepackage[noend]{algpseudocode}
\usepackage{xcolor}
\usepackage{booktabs, multirow}
\newcolumntype{P}[1]{>{\centering\arraybackslash}p{#1}}
\newcommand{\AAA}{\mathcal{A}}
\newcommand{\BBB}{\mathcal{P}}
\newcommand{\mpref}{\succ}
\newcommand{\lpref}{\prec}
\newcommand{\HR}{\sf HR}
\newcommand{\HRLQ}{\sf HRLQ}
\newcommand{\SMFQ}{\sf CCQ}
\newcommand{\SMFQUSET}{\sf MINSUM}
\newcommand{\CCQCC}{{\sf CCQ_{c1,c2}}}
\newcommand{\minmax}{\sf {MINMAX}}
\newcommand{\s}{a}
\newcommand{\cc}{p}
\newcommand{\NP}{\sf NP}
\newcommand{\Poly}{\sf P}
\newcommand{\OPT}{{\sf OPT}}
\newcommand{\SC}{\sf{SC}}
\newcommand{\MVC}{\sf{MVC}}

\newlist{strangenumerate}{enumerate}{2}

\setlist[strangenumerate]{noitemsep}
\setlist[strangenumerate,1]{label={\arabic*}}
\setlist[strangenumerate,2]{label={\arabic{strangenumeratei}.\arabic*}}

\newcounter{savedctr}
\begin{document}

\maketitle
\begin{abstract}
	We consider the problem of assigning agents to programs in the presence of two-sided
preferences, commonly known as the Hospital Residents problem. In the standard setting each
program has a rigid upper-quota which cannot be violated. Motivated by applications where
quotas are governed by resource availability, we propose and study the problem of computing optimal matchings with
cost-controlled quotas -- denoted as the ${\SMFQ}$ setting. 

In the $\SMFQ$ setting we have a cost associated with every program which denotes the cost of
	matching a single agent to the program. 
Our goal is to compute a matching that matches {\em all} agents,
	respects the preference lists of agents and programs and is optimal with respect to the cost criteria.
	We consider {\em envy-freeness} as a notion of optimality and study two optimization problems with respect to the costs -- minimize the total cost ($\SMFQUSET$) and minimize the maximum cost at a program ($\minmax$). We show that there is a sharp contrast
in the complexity status of these two problems -- $\minmax$ is polynomial time solvable whereas $\SMFQUSET$ is $\NP$-hard and hard to approximate within a constant factor unless $\Poly=\NP$ even under severe restrictions. On the positive side, we present approximation algorithms for the $\SMFQUSET$ for the general case and a special hard case. 
We achieve the approximation guarantee for the special case via a technically involved linear programming (LP) based algorithm. We remark
that our LP is for the general case of the problem.

\end{abstract}

\section{Introduction}
The problem of computing optimal matchings in the many-to-one setting and two-sided preferences
has been {extensively} investigated since it is applicable in several real-world applications like
assigning students to schools~\cite{AAST03} or elective courses~\cite{DBLP:conf/atal/OthmanSB10},
under-graduate students to university programs~\cite{BCCKP19}, medical interns (resident doctors)
to hospitals~\cite{Roth} and many more.
This setting, commonly called as the Hospital Residents ($\HR$) setting, is modelled as a bipartite graph $G = (\AAA \cup \BBB, E)$
where $\AAA$ and $\BBB$ denote the set of agents and programs respectively and 
$(a,p) \in E$ if and only if $a\in\AAA$ and $p\in\BBB$ are mutually acceptable to each other.
For a vertex $v$, let $\mathcal{N}({v})$ denote the vertices adjacent to $v$.
Each vertex $v \in \AAA \cup \BBB$ ranks the elements in $\mathcal{N}(v)$ in a  strict order and 
this ranking is called the {\em preference list} of that vertex.
If $x$ prefers $y$ over $z$, we denote it by $y \mpref_x z$.
In the standard setting a program $p$ has a positive
upper-quota $q(p)$. 

A {\em matching} $M \subseteq E$ in the $\HR$ instance $G$ is an assignment of agents to programs such that
each agent is matched to at most one program and a program is matched to at most $q(p)$ many agents.
Let $M(a)$ denote the program that agent $a$ is matched to in $M$ ($M(a) = \bot$ if $a$ is unmatched) and $M(p)$ denote the set of agents
matched to program $p$ in matching $M$. An agent $a$ prefers being matched to one of  the programs in $\mathcal{N}(a)$ over remaining unmatched.
A program $p$ is called {\em under-subscribed} in $M$ if $\mid\hspace{-0.1cm}M(p)\hspace{-0.1cm}\mid$ $< q(p)$.
{\em Stability} is a well-accepted notion of optimality in this setting and  is defined as follows.
\begin{definition}[Stable matchings in $\HR$ setting]
A pair $(a, p) \in E \setminus M$ is a blocking pair w.r.t. the matching $M$ if $p \mpref_a M(a)$ and $p$ is either under-subscribed in $M$ or there exists at least one agent $a' \in M(p)$ such that $a \mpref_p a'$. A matching $M$ is {\em stable} if there is no blocking pair w.r.t. $M$.
\end{definition}

It is well known that every $\HR$ instance 
admits a stable matching and all stable matchings are of the same size~\cite{GS62}.
Size of a matching plays a very important role in real-world applications where leaving agents unmatched is undesirable and sometimes even
unacceptable. 
In applications like school choice~\cite{AAST03} every child
must find a school. In case of matching sailors to billets in the US Navy~\cite{yang2003two,robards2001applying},
{ every} sailor must be assigned to some billet, apart from some additional constraints.
Relaxing stability in order to enable larger size matchings has been investigated in literature~\cite{huang2011popular,DBLP:journals/tcs/BiroMM10}. 
In all such works, there is an inherent assumption that the quotas of programs  
are {\em rigid} and cannot be compromised.

However, there are real-world applications where quotas 
are determined by logistic considerations like resource availability, classroom size and these may be flexible.
For instance, every semester elective allocation for under-graduate students at an educational institute
happens via an automated procedure and once the preferences 
are available to the academic office, course instructors are consulted to adjust class capacities if appropriate.
A recent work by Gajulapalli~et~al.~\cite{Vazirani} studies the
school choice problem in a two round-setting -- in the first round the quotas given as input are considered rigid. In 
 the second round though, some schools may increase their quota as suggested by the mechanism 
 in order to match {\em all} the students in a specific set in a {\em stability preserving} manner. 

Motivated by applications where $\AAA$-perfectness (matching {\em every} agent) is mandated  but is impossible to achieve in the presence of rigid quotas,
we introduce and study the setting where costs control {the number of agents matched to a program.
We denote this as the Cost Controlled Quota  ($\SMFQ$) setting.
An instance $H = (\AAA \cup \BBB, E)$ in the $\SMFQ$ setting is similar to an $\HR$ instance 
except that programs do not have input quotas, 
instead 
a program $p$ has a finite, non-negative integer cost $c(p)$ denoting
the cost of matching an agent to $p$.
A matching $M \subseteq E$ in the $\SMFQ$ setting is 
an assignment of agents to programs
where a program $p$ can be matched to as many agents as possible in $\mathcal{N}(p)$ 
and the number of agents assigned is controlled by costs.
If a program is not assigned any agents in a matching, we call the program as being {\em closed}.
Throughout, our goal is to compute an $\AAA$-perfect matching which is optimal
with respect to preferences as well as  costs.
We remark that a very recent work by Santhini~et~al.~\cite{Santhini} studies the
cost-based setting for matchings with one-sided preferences. However, the
problems considered in~\cite{Santhini} are very different from ours.

We first define the notion of optimality 
in the $\SMFQ$ setting w.r.t. the preferences and subsequently with respect to the costs.
Stability is inherently defined using the input quotas.
Envy-freeness~\cite{WR18}, a relaxation of stability, is 
a natural substitute since it is defined {\em without} input quotas. 
In absence of input quotas, envy-freeness is equivalent to stability.

\begin{definition}[Envy-free matchings]\label{def:efm}
Given a matching $M$, an agent $a$ has a {\em justified envy} (here onwards called envy) towards a matched agent $a'$, where $M(a') = p$ 
and $(a,p) \in E$ if $p \mpref_a M(a)$ and $a \mpref_p a'$.
The pair $(a, a')$ is an envy-pair w.r.t. $M$.  A matching $M$ is envy-free if there is no envy-pair w.r.t. $M$.
\end{definition}

Given a $\SMFQ$ instance, our goal is to
compute an $\AAA$-perfect envy-free matching 
subject to the following two optimization criteria with respect to costs.
\begin{enumerate}
	\item {\bf Minimize the total cost:} The total cost of a matching $M$ is defined as $\sum_{p \in \BBB} (\mid\hspace{-0.1cm}M(p)\hspace{-0.1cm}\mid$$\cdot c(p))$. Our goal is to compute an $\AAA$-perfect envy-free matching that minimizes the total cost -- we denote this as the $\SMFQUSET$ problem.
\item {\bf Minimize the maximum cost:} The maximum cost spent at a program for a matching $M$ is defined as $\max_{p \in \BBB} \{\mid\hspace{-0.1cm}M(p)\hspace{-0.1cm}\mid$$\cdot c(p)\}$.
Our goal is to compute an $\AAA$-perfect envy-free matching that minimizes the maximum cost -- we denote this  as the $\minmax$ problem.
\end{enumerate}
\begin{figure}[!ht]
\begin{minipage}{0.2\textwidth}
	\begin{align*}
		\s_1 &: \cc_1 \mpref \cc_0\\
		\s_2 &: \cc_1 \mpref \cc_0\\
		\s_3 &: \cc_1 \mpref \cc_0\\
		\s_4 &: \cc_1 \mpref \cc_2 \mpref \cc_0\\
		\s_5 &: \cc_2 \mpref \cc_3
	\end{align*}
\end{minipage}\hspace{1cm}
	\begin{minipage}{0.3\textwidth}
		\begin{align*}
			(0)\ \cc_0 &: \s_1 \mpref \s_2 \mpref \s_3 \mpref \s_4\\
			(1)\ \cc_1 &: \s_1 \mpref \s_2 \mpref \s_3 \mpref \s_4\\
			(6)\ \cc_2 &: \s_4 \mpref \s_5\\
			(11)\ \cc_3 &: \s_5
		\end{align*}
	\end{minipage}\hfill
	\begin{minipage}{0.3\textwidth}
		\begin{align*}
			M = \\ \{(a_1,p_1), (a_2,p_1),\\ (a_3,p_1), (a_4,p_1), \\(a_5,p_2)\}
		\end{align*}
	\end{minipage}
	\vspace{0.3cm}

	\caption{A $\SMFQ$ instance}
	\label{fig:challenges_1}
\end{figure}

Consider the  $\SMFQ$ instance $H$ in Fig.~\ref{fig:challenges_1} with five agents
and four programs. The preferences of the agents and the programs are as given in the figure. The number in the bracket beside a program denotes the cost of matching an agent to the respective program.
The matching $M$ shown in the figure is $\AAA$-perfect, as well as envy-free.
The total cost of $M$ is $10$ and the max-cost of $M$ is $6$. 
It is easy to verify that in the instance $H$,
the matching $M$ is an optimal solution for both the $\SMFQUSET$ as well as the $\minmax$ problem;
this need not be true in general.
The $\SMFQUSET$ can be viewed as a global objective whereas the $\minmax$ is a
local objective function.

\vspace{0.1in}
\noindent {\bf Our Contributions:} 
This is the first work that investigates the cost-controlled quotas
under two-sided preferences. We present an efficient algorithm for the
$\minmax$ problem whereas in a sharp contrast, the $\SMFQUSET$ turns out to be $\NP$-hard under
severe restrictions. To address the hardness, we present a novel  linear program for the $\SMFQUSET$ problem
and give primal dual approximation algorithm for a special case.
We state our results formally below.

\begin{theorem}
\label{thm:minmaxInP}
	The $\minmax$ problem is solvable in $O(m\log{m})$ time where $m = \mid\hspace{-0.1cm}E\hspace{-0.1cm}\mid$.
\end{theorem}
\begin{theorem}\label{thm:smfq_hardness}
The following hardness results hold for the  $\SMFQUSET$ problem.

\begin{enumerate}
	\item The $\SMFQUSET$ problem is strongly $\NP$-hard even when every agent has a preference list of length 
	exactly $f \geq 2$, there is a master list
	ordering on agents and programs and the instance has two distinct costs. \label{hpart1}
	\item The $\SMFQUSET$ problem cannot be approximated within a factor $\frac{7}{6}-\epsilon, \epsilon > 0$
	even when the instance has three distinct costs, unless $\Poly = \NP$.\label{hpart2}
	\end{enumerate}
\end{theorem}

We say that the instance has a {\em master list}~\cite{IRVING20082959} on agents if there exists a fixed ordering of agents such that the preference lists of every program obeys that ordering. A similar definition holds for a master list on programs.
In the $\SMFQ$ setting since $\AAA$-perfectness is guaranteed, it is natural for agents to submit  short preference lists. 
However, since there is a guarantee that every agent is matched, the central authority
is likely to impose a minimum requirement on the length of the preference list~\cite{SFAS}.
We further note that it is easy to compute the optimal solution for the $\SMFQUSET$  problem 
under the two extreme scenarios
-- when the preference lists are unit length or  when the preference lists are complete.

We complement our hardness results with the following approximation algorithms.
Let $\CCQCC$ denote the $\SMFQ$ instance with two distinct costs $c_1$ and $c_2$
such that $0 \leq c_1 < c_2$.
Let $\ell_a$ (respectively $\ell_p$) denote the length of the longest preference
list of an agent (respectively a program).

\begin{theorem}\label{thm:c1c2}
The $\SMFQUSET$ problem admits an $\ell_a$-approximation algorithm on $\CCQCC$ instances.
\end{theorem}
We achieve our approximation guarantee via a technically involved linear programming (LP) based algorithm and show that our analysis is tight. 
We remark that our LP is for the general $\SMFQUSET$ problem but the approximation guarantee is for the restricted setting. 
In many real-world applications, $\ell_a$ is typically small, in many cases a constant~\cite{IRVING2009213,10.1007/3-540-68530-8_32}. 

Finally, we present two simple approximation algorithms for $\SMFQUSET$ on general instances.
\begin{theorem} The $\SMFQUSET$ problem admits the following  approximation algorithms.
        \begin{enumerate} 
			\label{thm:smfq_approx}
                \item a linear time $\ell_p$-approximation algorithm. 
                        \label{part3}

		\item a $\mid\hspace{-0.1cm}\BBB\hspace{-0.1cm}\mid$-approximation algorithm\label{part2}.
        \end{enumerate}
\end{theorem}
The analysis of our $\ell_p$-approximation algorithm uses a natural lower-bound on the $\SMFQUSET$ problem. We also present a family of instances which shows that $\ell_p$ is the best guarantee that can be achieved using the particular lower-bound.
We establish 
that an optimal solution of the $\minmax$ problem also serves as a lower-bound for the $\SMFQUSET$ problem on the same instance and this gives us the $\mid\hspace{-0.1cm}\BBB\hspace{-0.1cm}\mid$-approximation algorithm.

\noindent{\bf Empirical  Evaluation.}
We implement our $\ell_p$-approximation and $\mid\hspace{-0.12cm}\BBB\hspace{-0.12cm}\mid$-approximation algorithms
presented in our paper for the general $\SMFQ$ instances for the $\SMFQUSET$ problem
and evaluate them on the $\SMFQ$ instances derived from real and synthetically generated $\HR$ instances.
For every algorithm, we report its execution time and the approximation guarantee obtained.
Additionally, we also measure the parameters that indicate the quality of the output matching.
Finally we report values for parameters that allows us to compare the $\SMFQ$ model and the $\HR$ model.

\vspace{0.1in}
\noindent{\bf Relation to other models. }
There are alternate ways to formulate an optimization problem in the $\SMFQ$ setting:
(i) given a total budget $\mathcal{B}$,
compute a largest envy-free matching with cost at most $\mathcal{B}$. 
(ii) given an $\HR$ instance and a cost for every program, {\em augment} the input quotas to compute
an $\AAA$-perfect envy-free matching with minimum total cost. 
The $\NP$-hardness for both these problems can be proved
by easily modifying the $\NP$-hardness reduction for $\SMFQUSET$.
As mentioned earlier, Gajulapalli~et~al.~\cite{Vazirani} consider the school choice problem
in a two-round setting. 
In the second round, their goal is to match {\em all} agents
in a particular set derived from the matching in the first round and they need to match them in an envy-free manner (called stability preserving in their work).
The $\SMFQ$ setting generalizes the matching problem in round-2 as follows:
If $c(p)$ denotes the cost of matching an agent to a program in round-2 then $\SMFQUSET$
computes an $\AAA$-perfect envy-free matching where the total cost of matching agents in round-2 is minimized.
Let $d(p)$ denote the deviation of program $p$, that is, the additional number of agents
matched to $p$ in round-2 beyond its input quota $q(p)$.
Then by setting $c(p) = 1$ for every program, the $\minmax$ problem computes an $\AAA$-perfect envy-free
matching in round-2 such that maximum deviation of a program is minimized.
We remark that in~\cite{Vazirani} the authors state that a variant of $\SMFQUSET$ problem (Problem~33, Section~7)
is $\NP$-hard. However, they do not investigate the problem in detail.

\noindent{\bf Other Related Work. }
For a review of results related to stable matchings in the $\HR$ setting,
we refer the reader to~\cite{GS62,GI89}.
Envy-freeness (also called as fairness) is a well-studied notion of optimality.
Structural properties of envy-free matchings in the $\HR$ setting are investigated in~\cite{WR18}. 
In the $\HRLQ$ setting ($\HR$ setting wherein hospitals have {\em lower-quotas})
envy-free matchings are studied in~\cite{Yokoi20,DBLP:SAGT20,FITUY15}.

We review the work where quotas are {\em flexible} or replaced by other constraints.
Flexible quotas in the college admission setting are studied in~\cite{smfqs}.
In their setting, no costs are involved but colleges may have ties in the preference lists and 
flexible quotas are used for tie-breaking at the last matched rank. 
In the student-project allocation setting, the problem of minimizing the maximum and total deviation from the initial target
is studied in~\cite{DBLP:phd/ethos/Cooper20}. 
A setting where courses make monetary transfers to students and have budget constraints is studied in~\cite{AAAI1817032}
in which they propose a new notion of optimality, namely {\em approximate stability}.
Funding constraints are 
studied in~\cite{aziz2020summer} in the context of allocating student interns to the projects funded by supervisors.

As mentioned earlier Santhini~et~al.~\cite{Santhini} consider a cost-based quota setting which is same as ours. However, their problems are in the one-sided 
preference list setting and they show an efficient algorithm to compute a min-cost matching of a desired {\em signature}. The signature allows to encode requirements about the number of agents matched to a particular rank. 
This result is in contrast to the hardness results we show for similar optimization problems
in this paper under the two-sided preference setting.

\vspace{0.2cm}
\noindent{\em Organization of the paper: }
In Section~\ref{sec:c1c2} we present a linear programming formulation for $\SMFQUSET$ problem. 
In Section~\ref{sec:lpalgo} we give an  LP based approximation algorithm for $\SMFQUSET$ on $\CCQCC$ instances. 
In Section~\ref{sec:ucalgo}, we present algorithmic results for $\minmax$ and $\SMFQUSET$ for general instances.
We present hardness 
results for the $\SMFQUSET$ problem
in Section~\ref{sec:uchardness}.
In Section~\ref{sec:exp} we present empirical results for our algorithms.
{ We discuss open questions in Section~\ref{sec:disc}  and conclude.}

\section{A Linear programming formulation for $\SMFQUSET$}\label{sec:c1c2}
In this section we present a linear program (LP) for the general $\SMFQUSET$ problem and 
discuss the challenges involved in designing a primal dual algorithm.

\subsection{Linear Program and its dual}
\begin{figure}[!ht]
{\footnotesize
\setlength\columnsep{55pt}
\noindent{\bf Primal: minimize}
\begin{equation}\label{eq:lp1}
	\sum\limits_{\cc\in \BBB}{c(\cc)\cdot\sum\limits_{(\s,\cc)\in E}{x_{\s,\cc}}}
\end{equation}
\noindent{\bf subject to}
\begin{equation}\label{eq:lp2}
		\sum_{\substack{p': \\p' = p\ \text{or} \\ \cc' \mpref_{\s'} \cc}}{x_{\s',\cc'}} \geq x_{\s,\cc}, \  \forall (\s',\cc)\in E, \s \lpref_{\cc} \s'
\end{equation}
\begin{equation}\label{eq:lp3}
	\sum\limits_{(\s,\cc)\in E}{x_{\s,\cc}} = 1,\ \ \ \forall \s \in \AAA
\end{equation}
\begin{equation}\label{eq:lp4}
	x_{\s,\cc} \geq 0,\ \ \ \forall (\s,\cc) \in E
\end{equation}
\\
\noindent{\bf Dual: maximize}
	\begin{equation}\label{eq:dualobj}
	\sum\limits_{a\in \AAA}{y_a}
\end{equation}
\noindent{\bf subject to}
\begin{equation}
\label{eq:lp5}
		y_a + \sum_{\substack{p': \\p' = p\ \text{or} \\ p' \lpref_{a} p}}\ \ {\sum\limits_{\substack{a':\\a' \lpref_{p'} a}}{z_{a,p',a'}}} - \sum\limits_{\substack{a':\\ a \lpref_p a'}}{z_{a',p,a}} \\\leq c(p),\ \  \ \ \forall (a,p)\in E
\end{equation}
	\begin{equation}\label{eq:lp7}
	z_{a',p,a} \geq 0,\ \ \ \forall (a',p) \in E, a \lpref_p a'
\end{equation}
}
\caption{Linear Program and its dual for the $\SMFQUSET$ problem}
\label{fig:lp-dual}
\end{figure}

Fig.~\ref{fig:lp-dual} shows the LP for the $\SMFQUSET$ problem. Let $H = (\AAA \cup \BBB, E)$ be a $\SMFQ$ instance.
Let $x_{a,p}$ be a primal variable for the edge $(a, p) \in E$: $x_{a, p}$ is $1$ if $a$ is matched to $p$, $0$ otherwise. The objective of the primal LP
(Eq.~\ref{eq:lp1}) is to minimize the total cost of all matched edges.
 Eq.~\ref{eq:lp2} encodes the envy-freeness constraint: if agent $a$  is matched to $p$ then every agent $a' \mpref_p a$ must be matched to either $p$ or a higher-preferred program than $p$, otherwise $a'$ envies $a$. In the primal LP, the envy-freeness constraint is present for a triplet $(a', p, a)$ where
 $a'\mpref_p a $. We call such a triplet a {\em valid} triplet. Eq.~\ref{eq:lp3} encodes $\AAA$-perfectness constraint.

\begin{figure}[!ht]
\centering
\includegraphics[scale=0.9]{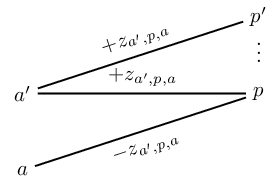}
	\vspace{0.3cm}
	\caption{Let $(a',p,a)$ be a valid triplet and
	$p' \mpref_{a'} p$. The edges shown in the figure are those whose dual constraint contains the variable $z_{a',p,a}$ in either positive or negative form.}
        \label{fig:envy_dual}
\end{figure}

In the dual LP, we have two kinds of variables, the $y$ variables which correspond to every agent
and the $z$ variables that correspond to every valid triplet in the primal program. The dual constraint
(Eq.~\ref{eq:lp5})
is for every edge $(a, p)$ in $E$.
The $y_a$ variable corresponding to an agent $a$
appears in the dual constraint corresponding to every edge incident on $a$.
The value $y_a$
can be interpreted as the cost paid by agent $a$ for matching $a$ to one of the programs in $\mathcal{N}(a)$.
For an edge $(a, p)$ and an agent $a' \mpref_p a$, the dual variable $z_{a', p, a}$ appears in negative form in exactly one constraint
and it is for the edge $(a, p)$. The same 
dual variable $z_{a',p,a}$ appears in positive form in the constraint for every edge $(a',p')$ such that $p'=p$ or $p' \mpref_{a'} p$ (refer Fig.~\ref{fig:envy_dual}).
	The value of $z_{a', p, a}$ can be interpreted as the cost paid by agent $a$ in matching $a'$ to a program $p'$
such that $p'=p$ or $p' \mpref_{a'} p$ to resolve potential envy-pair $(a',a)$ if $a$ gets matched to $p$.
Following are the useful facts about the linear program.

\noindent{\bf Fact 1.}
Let $a$ be a fixed agent.
If $y_a$ is incremented by a positive value $\Delta$ then it increments the left-hand side (lhs) of the 
dual constraint of {\em every} edge $(a,p)$ by $\Delta$ and it does not affect
the dual constraint of any edge incident on agent $a' \neq a$.\qed

\noindent{\bf Fact 2.}
Let $(a',p,a)$ be a fixed valid triplet.
If $z_{a',p,a}$ is incremented by a positive value $\Delta$
then it increments the lhs of the dual constraint of
{\em every} edge $(a',p')$ by $\Delta$ such that $p' = p$ or $p' \mpref_{a'} p$,
reduces the lhs of the dual constraint of {\em exactly one} edge $(a,p)$ by $\Delta$
and does not affect the dual constraint of any edge incident on agent $a'' \notin \{a,a'\}$.\qed

The following notation is used in illustrating the challenges
and in our $\ell_a$-approximation algorithm in the next section.
For a given dual setting and an edge,
if Eq.~\ref{eq:lp5} is satisfied with equality then we call such an edge as a {\em tight edge},
otherwise it is a {\em slack edge}. For an edge $(a,p)$, $slack(a,p)$ denotes its slack.
When referring to a $z$ variable, when a specific agent or program occurring in it does not matter, we use $\times$ in its place.
\begin{definition}[Threshold agent]
Let $M$ be a matching in the instance $H$. For every program $p$, $thresh(p)$ is the most-preferred agent $a$, if it exists,
such that $p \mpref_a M(a)$, otherwise $thresh(p)$ is $\bot$.
\end{definition}
The definition of threshold agent is similar to the threshold resident defined in \cite{MNNR18} and a barrier (vertex) defined in \cite{Vazirani}. 
We remark that the threshold agent depends on the matching $M$, hence when $M$ gets modified, the threshold agents for programs may change. 
\begin{definition}[Matchable edge]
For an envy-free matching $M$, and an agent $a$ (matched or unmatched),
we say that an edge $(a, p) \notin M$ is matchable if the dual constraint on $(a, p)$ is tight and $a = thresh(p)$,
	otherwise the edge is non-matchable.
\end{definition}
It is straightforward to verify that for an envy-free matching $M$, if we match agent $a$ along a matchable edge then the resultant matching
remains envy-free.

\subsection{Challenges}

A standard primal-dual approach for the $\SMFQUSET$ problem would be to begin with
 a dual feasible solution. The algorithm then repeatedly updates the dual till we obtain a primal feasible solution using the tight edges w.r.t. to
the dual setting.
We illustrate the challenges in using such an approach for the general $\SMFQUSET$ problem. 
Consider the $\SMFQUSET$ instance in Fig.~\ref{fig:challenges_1}.
Assume that we begin with an initial dual setting where all dual variables are set to $0$. 
The matching $M  = \{(a_1, p_0), (a_2, p_0), (a_3, p_0), (a_4,p_0)\}$ obtained on the tight edges
is envy-free but does not match agent $a_5$ and hence is not primal feasible.
Since no edge incident on $a_5$ is tight (slack on $(a_5, p_2)$ and $(a_5, p_3)$ is 6 and 11 respectively) we 
can set $y_5$ to 6 while maintaining dual feasibility. We observe that while this update makes the edge $(a_5, p_2)$
tight, adding the edge to the matching $M$ introduces an envy pair -- namely $a_4$ envying $a_5$. 
We note that this is the our first difficulty,
that is, while there are tight edges incident on an unmatched
agent, none of them may be matchable. 

The second difficulty stems from the following: in order to match $a_5$ along the (non-matchable) tight edge 
$(a_5,p_2)$ we must first resolve
the potential envy pair $(a_4, a_5)$, we must
{\em promote} agent $a_4$. With the current dual setting, $y_4$ cannot be increased hence a natural way is to
update a $z$ variable. This can indeed be achieved by setting $z_{{a_4}, {p_2}, {a_5}} = 1$, thus making 
$(a_4, p_1)$ tight. However, as encountered earlier, this edge is not 
matchable, since matching $a_4$ to
$p_1$ introduces several other envy pairs. Note that this chain of potential envy resolutions is triggered by the unamatched agent $a_5$. 
		Since, this chain can be arbitrarily long, several $z$ updates may be required.
		It is not immediate if
			these updates in $z$ variables can be charged to an update in some $y$ variable, thereby achieving a reasonable
		approximation ratio.

\section{An $\ell_a$-approximation algorithm for $\SMFQUSET$ on $\CCQCC$}\label{sec:lpalgo}
In this section we show that  when the $\SMFQ$ instance has only two distinct costs $c_1$ and $c_2$,
we are able to circumvent the challenges and obtain an $\ell_a$-approximation algorithm for the $\SMFQUSET$ problem.
We recall from Theorem~\ref{thm:smfq_hardness} that even in this restricted setting,
the problem remains {\sf NP}-hard.

\noindent {\bf High-level idea of the algorithm. }
Our LP based algorithm begins with an initial
feasible dual setting and an envy-free matching $M$ which need not be $\AAA$-perfect.
As long as $M$ is not $\AAA$-perfect, we pick an unmatched agent $a$ and increase the dual variable $y_a$.
We show that for an unmatched agent such an increase is possible and {\em all} edges incident on $a$ become tight due to the update.
However, none of the edges incident on $a$ may be matchable (since for every $p \in \mathcal{N}(a)$, $thresh(p) \neq a$). 
Under the restricted setting of two distinct costs we ensure that after a bounded number of updates to the 
$z$ variables,
at least one edge incident on $a$ is matchable.
Throughout we maintain the following invariants with respect to the matching $M$. 
\begin{itemize}
\item $M$ is envy-free, not necessarily $\AAA$-perfect and every matched edge is tight.
\item For an agent $a$ (matched or unmatched), for every $p \mpref_{a} M(a)$, 
	either (i) $(a,p)$
		is tight and $thresh(p) \neq a$ or (ii) $slack(a,p) = c_2-c_1$.
\end{itemize}
We remark that when the matching is modified, thresholds may change,
due to which a tight, non-matchable edge may become matchable.
As long as there exists such an edge, we match it.
This is achieved by the {\em free-promotions} routine. 
The free-promotions routine checks if there exists a matchable edge $(a,p)$.
If there is no such edge, the routine terminates.
Otherwise, it matches $(a,p)$, re-computes the threshold agents and repeats the search.
Checking for a matchable edge and computing threshold agents takes
$O(m)$ time where $m$ is the number of edges in the underlying graph. Since, no agent is demoted in this process, 
the free-promotions routine runs in $O(m^2)$ time.

\noindent{\bf Description of the algorithm. }
Algorithm~\ref{algo:dualalgo_2} gives the pseudo-code.
In Fig.~\ref{fig:exinst} we give an illustrative example 
which depicts the 
key steps of the algorithm on a $\CCQCC$ instance. 
We begin with an empty matching $M$ and by setting all $y$ variables to $c_1$ and all $z$ variables to $0$ (line~\ref{line:init1}). Following this, for every
agent $a$ with a cost $c_1$ program in $\mathcal{N}(a)$
we match the agent to its most-preferred program with cost $c_1$ ({\bf for} loop at line~\ref{line:forloop}).
Next, we compute the threshold agent for every program w.r.t. $M$.
As long as $M$ is not $\AAA$-perfect, we pick an arbitrary unmatched agent $a$ and update the dual variables as follows.
	
\begin{algorithm}[!ht]
	\begin{algorithmic}[1]
		\State let $M = \emptyset$, all $y$ variables are set to $c_1$ and all $z$ variables are set to $0$\label{line:init1}
		\For {every agent $a \in \AAA$ s.t. $\exists p \in \mathcal{N}(a)$ such that $c(p) = c_1$}\label{line:forloop}
		    \State let $p$ be the most-preferred program in $\mathcal{N}(a)$ s.t. $c(p) = c_1$ and let
		    $M = M \cup \{(a,p)\}$\label{line:c1p}
		\EndFor
		\State compute $thresh(p)$ for every program $p\in \BBB$\label{line:beforeloop}
		\While {$M$ is not $\AAA$-perfect}\label{line:loop1a2}
			\State let $a$ be an unmatched agent\label{line:picka}
			\While {$a$ is unmatched}\label{line:loop_a}
				\State set $y_a = y_a + c_2 - c_1$ \label{line:yra2}
				\If {there exists a matchable edge incident on $a$}
					\State $M = M \cup \{(a,p) \mid (a,p)$ is the most-preferred matchable edge for  $a\}$\label{line:lineif}
						\State perform free-promotions routine and re-compute thresholds
				\Else
					\State $\BBB(a) = \{ p \in \mathcal{N}(a) \mid p \mpref_a M(a), (a,p)$
						is tight  and $thresh(p) \neq a\}$\label{line:ba}
					\While {$\BBB(a) \neq \emptyset$}\label{line:loop2a2}
						\State let $a'$ be the threshold agent of some program in $\BBB(a)$  \label{line:rprimea2}
						\State let $\BBB(a,a')$ denote the set of programs in $\BBB(a)$ whose threshold agent is $a'$\label{line:paaprime}
						\State let $p$ be the least-preferred program for $a'$ in $\BBB(a, a')$  \label{line:pickp}
						\State set $z_{a',p,a} = c_2-c_1$ \label{line:za2}
						\State let $(a',p')$ be the most-preferred matchable edge incident on $a'$.
						Unmatch $a'$ if matched and
						let $M = M \cup \{(a', p')\}$ \label{line:pro1a2}
	    			    		\State execute free-promotions routine, re-compute thresholds and the set $\BBB(a)$\label{line:bar}
		        		\EndWhile
				\EndIf
			\EndWhile
		\EndWhile
		\State return $M$
	\end{algorithmic}
	\caption{Algorithm to compute an $\ell_a$-approximation of $\SMFQUSET$ on $\CCQCC$}
	\label{algo:dualalgo_2}
\end{algorithm}

\begin{strangenumerate}
    \item For the agent $a$, we increase $y_a$ by $c_2-c_1$. We ensure that
	    the dual setting is feasible 
		and all edges incident on $a$ become tight for the dual constraint in Eq.~\ref{eq:lp5}. 
		\label{step:atight}
 Although this step makes all edges incident on $a$ tight, they may not be necessarily matchable. Recall
that a tight edge $(a, p)$ is matchable if $thresh(p) = a$. 
    \item If there is a program $p$ such that $(a, p)$ is matchable, then $a$ is immediately matched to the most-preferred such program $p$ (line~\ref{line:lineif})
    and we are done with matching agent $a$.
		   Since the matching is modified, we execute free-promotions routine.
    \item In case there is no such program for which $a$ is the threshold agent,
	    we update carefully selected $z$ variables in order to either promote the threshold agent (if matched) or
		match the (unmatched) threshold agent via the following steps.
		\begin{strangenumerate}
		    \item We compute the set $\BBB(a)$ of programs $p \in \mathcal{N}(a)$
			    such that the dual constraint on edge $(a,p)$ is tight and
			    $thresh(p) \neq a$ and $p \mpref_a M(a)$ (line~\ref{line:ba}).
			    In other words, $\BBB(a)$ is the set of programs in the neighbourhood of $a$
				such that $p$ is higher-preferred over $M(a)$ and edge $(a,p)$ is tight but not matchable.
			    \label{step:ba}
		    \item By the definition of $\BBB(a)$, for every $p_j \in \BBB(a)$,
			    there exists $thresh(p_j) = a' \neq a$.
				We pick an arbitrary agent $a'$ that is a threshold of some program in $\BBB(a)$ (line~\ref{line:rprimea2}).
			    Note that the agent $a'$ can be the threshold agent of more than one programs in $\BBB(a)$,
				and we let $\BBB(a,a')$ denote the set of programs in $\BBB(a)$ for whom $a'$ is the threshold.
				Let $p$ be the least-preferred program for $a'$ in $\BBB(a,a')$ (line~\ref{line:pickp}).
				\label{step:pickp}
			\item 	Our goal is to match $a'$ to a program $p'$ such that $p' = p$ or $p' \mpref_{a'} p$.
				By the choice of $a, a'$ and $p$ and from the primal LP, 
				$(a',p,a)$ is a valid triplet and therefore there exists a 
				dual variable $z_{a',p,a}$ (refer Fig.~\ref{fig:envy_dual}).
				We set $z_{a',p,a}$ to $c_2-c_1$ (line~\ref{line:za2}).
				We ensure that this update maintains dual feasibility.
			    \label{step:zupdt}
		    \item 
			    Recall that the variable $z_{a',p,a}$ appears in the positive form in the dual
			    constraint of every edge $(a',p')$ such that $p' = p$ or $p' \mpref_{a'} p$.
		    We ensure that this update results in making all edges $(a',p')$ tight
				and at least one of these 
				becomes matchable.
				We match $a'$ along the most-preferred matchable edge (line~\ref{line:pro1a2}).
				Recall that $z_{a',p,a}$ variable appears in negative form in the dual constraint of edge $(a,p)$,
				hence edge $(a,p)$ becomes slack after this update.
			    \label{step:matchaprime}
		    \item Since $M$ is modified, we execute free-promotions routine.
			    If a tight edge incident on $a$ becomes matchable, then $a$ is matched inside the free-promotions routine. \label{step:fp}
		    \item We remark that the set $\BBB(a)$ computed in line~\ref{line:ba} is dependent on the matching $M$, specifically $M(a)$ and the threshold agents w.r.t. $M$. In order to maintain a specific slack value on the edges that is useful in maintaining dual feasibility and ensuring progress, we re-compute the set $\BBB(a)$ (line~\ref{line:bar}) and re-enter the loop in line~\ref{line:loop2a2} if $\BBB(a) \neq \emptyset$. \label{step:recompBa}
		\end{strangenumerate}
\end{strangenumerate}
\begin{figure}[!ht]
		\begin{center}
    \begin{minipage}{0.2\textwidth}
    \begin{align*}
          a_1 &: p_1\mpref p_2\mpref p_0\\
          a_2 &: p_2\mpref p_3\mpref p_0\\
          a_3 &: p_1\mpref p_2\mpref p_3\\
	  \\[-10pt]
	  \hline
	  \\[-10pt]
                  (0)\ p_0 &: a_1\mpref a_2\\
                  (1)\ p_1 &: a_1\mpref a_3\\
                  (1)\ p_2 &: a_1\mpref a_2\mpref a_3\\
                  (1)\ p_3 &: a_2\mpref a_3
                  \end{align*}
	  \end{minipage}\hfill
          \begin{minipage}{0.65\textwidth}
		  \vspace{0.5cm}
		\begin{itemize}
            \item $M = \{(a_1,p_0), (a_2,p_0)\}$
			\item \textcolor{blue}{(\ref{step:atight})} $a = a_3$, $y_{a_3} = 1$, tight edges on $a_3$ are $\{(a_3,p_1)$, $(a_3,p_2)$, $(a_3,p_3)\}$, $thresh(p_1) = thresh(p_2) = a_1$ and $thresh(p_3) = a_2$
			\item \textcolor{blue}{(\ref{step:ba})} $\BBB(a_3) = \{p_1, p_2, p_3\}$
			\item \textcolor{blue}{(\ref{step:pickp},\ref{step:zupdt})} let $a' = a_1$, then $p = p_2$, $z_{a_1,p_2,a_3} = 1$
			\item \textcolor{blue}{(\ref{step:matchaprime})} Tight edges on $a_1$ are $\{(a_1, p_1), (a_1, p_2)\}$, $p' = p_1$, $M = \{(a_1,p_1), (a_2,p_0)\}$, tight edges on $a_3$ are $\{(a_3,p_1), (a_3,p_3)\}$
			\item \textcolor{blue}{(\ref{step:fp})} $thresh(p_1) = a_3$, $M = \{(a_1,p_1), (a_2,p_0), (a_3,p_1)\}$
			\item \textcolor{blue}{(\ref{step:recompBa})} $\BBB(a_3) = \emptyset$
		\end{itemize}
                  \end{minipage}
			\vspace{0.3cm}
			\caption{A $\CCQCC$ instance.
			An execution of Algorithm~\ref{algo:dualalgo_2} is illustrated by giving the state of the algorithm. The blue numbers in the bracket correspond to the labels of steps mentioned in the description.}
        \label{fig:exinst}
		\end{center}
	\end{figure}

\noindent {\bf Observations. }
We observe the following properties of the algorithm.

\noindent{\bf (P1)} At line~\ref{line:beforeloop}, no agent is assigned to any program with cost $c_2$ and 
    for every agent $a$ (matched or unmatched), every program $p \mpref_a M(a)$ has cost $c_2$. \qed

    Next we observe that whenever a matched agent $a$ changes its partner from $M(a)$ to program $p$, 
    we have $thresh(p) = a$. By the definition of the threshold agent, $p \mpref_a M(a)$,
    which implies the following.

\noindent{\bf (P2)}
    A matched agent never gets demoted. \qed

{\bf Fact 1} and {\bf Fact 2} together imply that during the execution of
the algorithm, the only edge that can become slack is the edge $(a,p)$ in line~\ref{line:za2}.
Note that $a$ is an unmatched agent. 
Therefore no tight edge incident on a matched agent 
can become slack, implying the following.

\noindent {\bf (P3)} A tight edge incident on a matched agent always remains tight. \qed

We also observe that only a matchable edge is matched throughout the algorithm.
This implies that the edge is tight when matched.
By {\bf (P3)}, a matched edge (being incident on a matched agent) always remain tight,
implying the following.

\noindent {\bf (P4)} All matched edges are tight at the end of the algorithm. \qed

\subsection{Proof of correctness}
We first prove that matching $M$ is envy-free.
\begin{lemma}\label{lem:mefm}
Matching $M$ is envy-free throughout the execution of the algorithm.
\end{lemma}
\begin{proof}
	Matching $M$ is trivially envy-free after line~\ref{line:init1}.
	Any two agents $a$ and $a'$ that are matched in line~\ref{line:c1p}
	are matched to a program with cost $c_1$ and by the choice made in line~\ref{line:c1p},
	it is clear that they do not form an envy-pair.
	By {\bf (P1)}, 
	every unmatched agent $a$ has only cost $c_2$ programs in $\mathcal{N}(a)$
	thus, no unmatched agent envies an agent matched in line~\ref{line:c1p}.
	Thus, $M$ is envy-free before entering the loop at line~\ref{line:loop1a2}.

	Suppose $M$ is envy-free before a modification in $M$ inside the loop.
	We show that it remains envy-free after the modification.
	Matching $M$ is modified either at line~\ref{line:lineif} or line~\ref{line:pro1a2}
	or inside the free-promotions routine. 
	In all these places, only a matchable edge $(a_i,p_j)$ is matched.
	Therefore no agent $a' \neq a_i$ envies $a_i$ after this modification.
	Before this modification $a_i$ did not envy $a' \neq a_i$
	and by {\bf (P2)} $a_i$ (if matched) is not demoted, therefore
	$a_i$ does not envy $a' \neq a_i$ 
	after the modification. Thus, $M$ remains envy-free.
\end{proof}

Next we make the following observation about the innermost {\bf while} loop (line~\ref{line:loop2a2}).

\begin{claim}\label{cl:boundp}
    Let $a$ be a fixed unmatched agent selected in line~\ref{line:picka}
	and consider an iteration of the loop at line~\ref{line:loop_a} during which
	the algorithm enters {\bf else} part.
	Suppose during an iteration of the loop at line~\ref{line:loop2a2},
	for some $p_k \in \mathcal{N}(a)$, $p = p_k$ is selected at line~\ref{line:pickp}.
	Then at the end of iteration, $slack(a,p_k) = c_2-c_1$ and 
	$p \neq p_k$ during subsequent iterations of the loop.
	Therefore, at most $\ell_a$ many distinct $z_{\times,p_k,a}$ variables are
	updated during the iteration of the loop at line~\ref{line:loop_a}.
\end{claim}
\begin{proof}
By the choice of $p_k$, 
the edge $(a,p_k)$ was tight before this iteration.
	By {\bf Fact 2}, the update on $z_{\times,p_k,a}$ reduces the lhs of the dual constraint of
the edge $(a,p_k)$ by $c_2-c_1$.
Thus, after this update, $slack(a,p_k) = c_2-c_1$.
Therefore, when $\BBB(a)$ is re-computed at line~\ref{line:bar}, $p_k \notin \BBB(a)$.
	Also observe that no other dual update in $z_{\times,p_j,a}$ inside the loop at line~\ref{line:loop2a2}
	for $p_j \neq p_k$
	affects the slack of edge $(a,p_k)$.
Thus, in a subsequent iteration of this loop, $p_k$ is never selected as $p$ again.

	For every $p_k$ selected as $p$ in line~\ref{line:pickp},
	a distinct $z_{\times,p_k,a}$ variable is updated.
	Thus, there are at most $\mid\hspace{-0.1cm}\BBB(a)\hspace{-0.1cm}\mid$ many distinct $z_{\times,p_k,a}$ variables are
	updated inside the loop at line~\ref{line:loop2a2} in an iteration of the loop at line~\ref{line:loop_a}.
	By observing that $\BBB(a) \subseteq \mathcal{N}(a)$, we get $\mid\hspace{-0.1cm}\BBB(a)\hspace{-0.1cm}\mid \leq \ell_a$, hence the claim follows.
\end{proof}

Now, we proceed to show that the dual setting is feasible and that
the algorithm terminates in polynomial time.
Recall that if edge $(\hat{a},\hat{p})$ is non-matchable then
either $(\hat{a},\hat{p})$ is slack or $thresh(\hat{p}) \neq \hat{a}$.
In our algorithm, we maintain a stronger invariant: 
for every agent $a$ and for every program $p$ higher-preferred 
over $M(a)$, we maintain that either {\em all} non-matchable edges $(a,p)$ are slack or 
	all non-matchable edges $(a,p)$ are tight and
for {\em every} edge,
$thresh(p) \neq a$.
Moreover, we also maintain a specific slack value when the edges are slack.
We categorize agents based on these two cases (see Fig.~\ref{fig:type1} and Fig.~\ref{fig:type2}).
\begin{definition}[type-1 and type-2 agents]
An agent $a$ is called a {\em type-1} agent if for every program $p \mpref_{a} M(a)$,
$slack(a,p) = c_2-c_1$. 
An agent $a$ is called a {\em type-2} agent if $a$ is matched and
for every program $p \mpref_{a} M(a)$,
$slack(a,p) = 0$ and $thresh(p) \neq a$.
\end{definition}

\begin{figure}[!ht]
	\centering
	\includegraphics[scale=0.9]{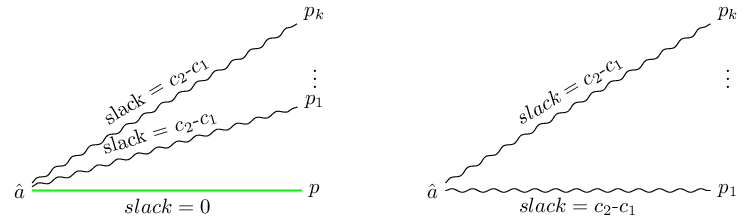}
		\caption{Type-1 agent $\hat{a}$: if matched, then $p = M(\hat{a})$ and $\forall p_j \mpref_{\hat{a}} p$ otherwise, $\forall p_j \in \mathcal{N}(\hat{a})$}
		\label{fig:type1}
	\end{figure}

\begin{figure}[!ht]
	\centering
	\includegraphics[scale=0.9]{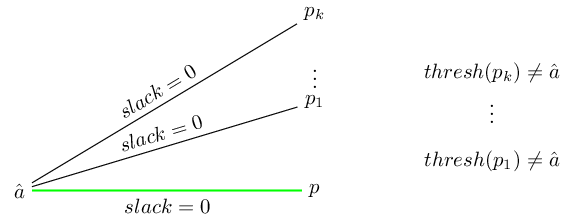}
		\caption{Type-2 agent $\hat{a}$: always matched, $p = M(\hat{a})$ and $\forall p_j \mpref_{\hat{a}} p$}
        \label{fig:type2}
\end{figure}

We remark that type-1 agent could be either matched or unmatched but type-2 agent is always matched.
Recall that if $a' = a_j$ is unmatched then $M(a_j) = \bot$ and therefore, every program $p_j \in \mathcal{N}(a_j)$
satisfies the condition that $p_j \mpref_{a_j} M(a_j) = \bot$.
We claim that a type-1 agent is selected as $a'$ at most once inside the loop at line~\ref{line:loop2a2}.

\begin{claim}\label{cl:boundaprime}
	Let $a_j$ be a type-1 agent such that $a' = a_j$ is selected in an arbitrary iteration of the loop
	at line~\ref{line:loop2a2}. Then, at the termination of the loop, $a_j$ is a type-2 agent
	and in subsequent iterations of the loop, $a' \neq a_j$.
\end{claim}
\begin{proof}
Since $a_j$ is a type-1 agent, for every program $p_j \mpref_{\hat{a}} M(\hat{a})$,
$slack(a_j, p_j) = c_2-c_1$.
Suppose $p = p_k$ is selected in line~\ref{line:pickp}.
Then by {\bf Fact 2}, for every $p_t$ such that $p_t = p_k$ or $p_t \mpref_{a_j} p_k$,
the dual update in line~\ref{line:za2} results in making all $(a_j,p_t)$ edges tight.
Also, since $thresh(p_k) = a_j$, at least one of these newly tight edges (specifically, $(a_j, p_k)$)
becomes matchable.
Therefore, $M(a_j)$ is modified inside the iteration (line~\ref{line:pro1a2}),
implying that $a_j$ is either matched or promoted.
	The choice of $M(a_j)$, that is, $p'$ in line~\ref{line:pro1a2}
	is such that for every $p_j \mpref_{a_j} M(a_j) = p'$, the edge $(a_j,p_j)$ is tight and $thresh(p_j) \neq a_j$.
	Thus, when the iteration ends, $a_j$ is a type-2 agent.

	By {\bf (P3)}, the tight edges incident on $a_j$ remain tight throughout the algorithm.
	In subsequent iterations, agent $a_j$ may further get promoted by the free-promotions routine 
such that for every $p_j \mpref_{a_j} M(a_j)$, $slack(a_j,p_j) = 0$ and $thresh(p_j) \neq a_j$.
Therefore, $a_j$ remains a type-2 agent in all subsequent iterations of the loop.
	This implies that $a_j$ is not the threshold for any program $p_j \mpref_{a_j} M(a_j)$,
	in particular for any program $p_j \in \mathcal{N}(a)$ for the chosen $a$.
Thus, during subsequent iterations of the loop,
$a' \neq a_j$.
\end{proof}

In Lemma~\ref{lem:c1_unm_cond},
we establish that at a specific step during the algorithm, every agent is either type-1 or type-2.
This property is crucial in showing dual feasibility and termination.

\begin{lemma}\label{lem:c1_unm_cond}
	Before every iteration of the loop starting at line~\ref{line:loop_a}, an agent $\hat{a}$ is either
	a type-1 agent or a type-2 agent.
\end{lemma}
\begin{proof}
We prove this by induction.
Before the first iteration of the loop at line~\ref{line:loop_a},
	suppose agent $\hat{a}$ is matched. Then {\bf (P1)} and the initial dual setting together imply that
	for every program $p_j \mpref_{\hat{a}} M(\hat{a})$, $slack(\hat{a},p_j) = c_2-c_1$.
	Therefore $\hat{a}$ is a matched type-1 agent.
	Suppose $\hat{a}$ is unmatched. Then, by {\bf (P1)}, every program $p_j \in \mathcal{N}(\hat{a})$,
	$c(p_j) = c_2$, therefore the initial dual setting implies that $slack(\hat{a}, p_j) = c_2-c_1$.
	This implies that $\hat{a}$ is an unmatched type-1 agent.

	Consider an arbitrary agent $\hat{a}$.
	Suppose that $\hat{a}$ is either type-1 or type-2
	before $l$-th iteration of the loop.
	It is clear that $a$ selected in line~\ref{line:picka} is different than $a'$ selected at line~\ref{line:rprimea2}.
	During the $l$-th iteration, either $a = \hat{a}$ in line~\ref{line:picka}
	or $a' = \hat{a}$ in line~\ref{line:rprimea2} or $\hat{a}$ is promoted inside the free-promotions routine.
	We show that in each of the cases, $\hat{a}$ is either type-1
	or type-2 before $(l+1)$-th iteration begins.

        \begin{enumerate}[label={(\roman*)}]
	\item {\bf $a = \hat{a}$ in line~\ref{line:picka}: }
		It implies that $\hat{a}$ is unmatched.
		By induction hypothesis, $\hat{a}$ is a type-1 agent, therefore
		for every $p_j \in \mathcal{N}(\hat{a})$,
			$slack(\hat{a}, p_j) = c_2-c_1$.
			Then, the update in line~\ref{line:yra2} results in making all edges incident on $\hat{a}$ tight.
			We consider the following two cases -- $\hat{a}$ remains unmatched during
			the $l$-th iteration or $\hat{a}$ gets matched.
			\begin{itemize}
				\item {\bf $\hat{a}$ remains unmatched during the $l$-th iteration: }
			Then the while loop at line~\ref{line:loop2a2} must have been executed.
			During an iteration of the loop at line~\ref{line:loop2a2},
			if $p = p_j$ then the slack of the edge $(\hat{a}, p_j)$ becomes $c_2-c_1$
					after the dual update in line~\ref{line:za2} (by {\bf Fact 2}).
			We show that for every $p_j \in \mathcal{N}(\hat{a})$, 
			there is some iteration of the loop at line~\ref{line:loop2a2}
			such that $p = p_j$ is selected, thereby implying that when the loop terminates,
			for every edge $(\hat{a}, p_j)$, slack becomes $c_2-c_1$.
					Once this is shown, it is clear that before the $(l+1)$-th iteration,
					$\hat{a}$ is a type-1 agent.

			Suppose for contradiction that for some program $p_j$, $p = p_j$ is never selected.
			Since the edge $(\hat{a}, p_j)$ is tight before the loop execution began, it must be the case
			that either $p_j \lpref_{\hat{a}} M(\hat{a})$ or $thresh(p_j) = \hat{a}$.
			The first case implies that $M(\hat{a}) \neq \bot$, a contradiction that $\hat{a}$ remains unmatched during
			the $l$-th iteration. In the second case, since $thresh(p_j) = \hat{a}$, the edge $(\hat{a},p_j)$
			was matchable inside the free-promotions routine, thus $\hat{a}$ must have been matched inside the free-promotions routine, leading to a contradiction again.
			Thus, for every $p_j \in \mathcal{N}(\hat{a})$, there is some iteration of the loop during which
					$p = p_j$. This implies that when the loop at line~\ref{line:loop2a2}
					terminates, for every $p_j \in \mathcal{N}(\hat{a})$,
					$slack(\hat{a}, p_j) = c_2-c_1$.

		\item {\bf $\hat{a}$ gets matched during the $l$-th iteration: }
			Recall that all edges incident on $\hat{a}$ are tight after the dual update in line~\ref{line:yra2}.
			If $\hat{a}$ is matched at line~\ref{line:lineif} then 
			the $l$-th iteration immediately terminates. Thus, before the $(l+1)$-th iteration, 
					for every $p_j \mpref_{\hat{a}} M(\hat{a})$, $slack(\hat{a}, p_j) = 0$
					and by the choice made in line~\ref{line:lineif}, $thresh(p_j) \neq \hat{a}$,
					implying that $\hat{a}$ is a type-2 agent.

			If $\hat{a}$ is matched inside the loop at line~\ref{line:loop2a2} then the free-promotions
			routine must have matched it.
					Consider the last iteration of the loop at line~\ref{line:loop2a2}
					during which the free-promotions routine matched or promoted $\hat{a}$ 
					and let $M(\hat{a}) = p_t$.
					We will show that for $p_j \mpref_{\hat{a}} p_t$, $slack(a_j,p_j) = c_2-c_1$,
					thereby implying that $\hat{a}$ is a matched type-1 agent
					before $(l+1)$-th iteration begins.

					By Claim~\ref{cl:boundp}, it is enough to show that for every $p_j \mpref_{\hat{a}} p_t$,
					$p = p_j$ is chosen is some iteration of the loop at line~\ref{line:loop2a2}.
					Suppose not. Then, there exists some $p_j$ such that $(\hat{a}, p_j)$ is tight
					after the loop at line~\ref{line:loop2a2} terminates.
					By the choice of $p_t$ inside the free-promotions routine, 
					$(\hat{a}, p_j)$ was non-matchable, implying that $thresh(p_j) \neq \hat{a}$.
					Hence during the last iteration of the loop, when $\BBB(\hat{a})$ was re-computed
					in line~\ref{line:bar}, $p_j \in \BBB(\hat{a})$, that is, $\BBB(\hat{a}) \neq \emptyset$.
					This contradicts that the loop terminated after this iteration.
					Therefore, for every $p_j \mpref_{\hat{a}} p_t$, $p_j$ was selected in some iteration of the loop at line~\ref{line:loop2a2}, thereby implying that before the $(l+1)$-th iteration of the loop at line~\ref{line:loop_a},
					$\hat{a}$ is a matched type-1 agent.
			\end{itemize}

		\item {\bf $a' = \hat{a}$ at line~\ref{line:rprimea2}}:
		Consider the first iteration of the loop at line~\ref{line:loop2a2} when this happens.
		Note that the dual update in line~\ref{line:yra2} does not affect the slack on edges incident on $\hat{a}$.
		Since $\hat{a}$ is a threshold for some program $p_j \mpref_{\hat{a}} M(\hat{a})$,
		by the induction hypothesis, $\hat{a}$ is a type-1 agent.
		Therefore, for every $p_j \mpref_{\hat{a}} M(\hat{a})$, $slack(\hat{a},p_j) = c_2-c_1$ 
		before this iteration of the loop at line~\ref{line:loop2a2}.
		By Claim~\ref{cl:boundaprime}, $\hat{a}$ is a type-2 agent when the loop terminates.
		Therefore when $(l+1)$-th iteration of the loop at line~\ref{line:loop_a} begins,
		$\hat{a}$ is a type-2 agent.

	\item {\bf $a \neq \hat{a}$ and $a' \neq \hat{a}$ but $\hat{a}$ is promoted inside the free-promotions routine}:
		First note that none of the dual updates in the $l$-th iteration affect any edge incident on $\hat{a}$.
		Thus, if $\hat{a}$ is promoted inside the free-promotions routine, 
		then by the induction hypothesis, $\hat{a}$ must be a type-2 agent.
		Thus, for every $p_j \mpref_{\hat{a}} M(\hat{a})$, $slack(\hat{a}, p_j) = 0$ and $thresh(p_j) \neq \hat{a}$
		and some update in the matching must have made one of these edges matchable, that is,
		for some tight edge $(\hat{a},p_j)$, $thresh(p_j) = \hat{a}$.
		Consider the last iteration of the loop at line~\ref{line:loop2a2} when the
		free-promotions routine promoted $\hat{a}$.
		Then, by the choice of $M(\hat{a})$ inside the routine, for every program $p_j \mpref_{\hat{a}} M(\hat{a})$,
		edge $(\hat{a}, p_j)$ is non-matchable. This implies that for every such $p_j$,
		$thresh(p_j) \neq \hat{a}$. Thus, $\hat{a}$ remains a type-2 agent
		when the $(l+1)$-th iteration begins.
\end{enumerate}
This completes the proof of the lemma.
\end{proof}

Next, we show that the dual setting is feasible.

\begin{lemma}\label{lem:dualfeas}
	The dual setting is feasible throughout the algorithm.
\end{lemma}
\begin{proof}
It is clear that the dual setting is feasible before entering the loop after line~\ref{line:loop_a}
	for the first time.
	We show that if the dual setting is feasible
	before an arbitrary dual update (either line~\ref{line:yra2} or line~\ref{line:za2})
	then it remains feasible after the update.

\begin{itemize}
	\item {\bf Update at line~\ref{line:yra2}:} 
		Since $a$ is unmatched, by Lemma~\ref{lem:c1_unm_cond},
		$a$ is a type-1 agent and therefore, the slack on every edge 
		incident on $a$ is $c_2-c_1$.
		By {\bf Fact 1}, this update increases the lhs of
		every edge incident on $a$ by $c_2-c_1$ and the iteration
		of the loop at line~\ref{line:loop_a} terminates.
		Therefore the dual setting is feasible.
	\item {\bf Update at line~\ref{line:za2}:}
		We note that the update in line~\ref{line:za2} increases
		the lhs of a subset of edges incident on agent $a'$ (by {\bf Fact 2}).
		Therefore we show that for an arbitrary agent $a_j$ selected
		as $a'$, the dual setting on the affected edges is feasible
		after the update.

		Consider the first iteration of the loop at line~\ref{line:loop2a2}
		wherein an arbitrary $a_j$ is selected as $a'$ in line~\ref{line:rprimea2}.
		Since $a \neq a' = a_j$, the type of $a_j$ before 
		execution of the loop at line~\ref{line:loop2a2} began is same
		as its type before entering the loop at line~\ref{line:loop_a}.
		Suppose $a_j$ is a type-2 agent then the fact that 
		$a_j$ is threshold at some program in $\BBB(a)$ contradicts 
		that for every program $p_j \mpref_{a_j} M(a_j)$, $thresh(p_j) \neq a_j$.
		Therefore, $a_j$ is a type-1 agent.
		This implies that for every $p_j \mpref_{a_j} M(a_j)$, the slack 
		of the edge $(a_j,p_j)$ is $c_2-c_1$, therefore the dual update in
		line~\ref{line:za2} maintains dual feasibility.
		By Claim~\ref{cl:boundaprime}, this is the only
		iteration of the loop at line~\ref{line:loop2a2}
		when $a' = a_j$.
		Therefore, when the execution of 
		loop at line~\ref{line:loop2a2} terminates (followed by immediate termination of 
		the loop at line~\ref{line:loop_a}), the dual setting remains feasible.
\end{itemize}
This completes the proof of the lemma.
\end{proof}

Now, we show that the algorithm terminates in polynomial time
and computes an $\AAA$-perfect matching $M$.
\begin{lemma}\label{lem:terminates}
	Algorithm~\ref{algo:dualalgo_2} terminates by computing an $\AAA$-perfect matching in polynomial time.
\end{lemma}
\begin{proof}
	We first show that in every iteration of the loop in line~\ref{line:loop_a},
	either an unmatched agent is matched or
	at least one agent is promoted:
	by Lemma~\ref{lem:c1_unm_cond} and {\bf Fact 1},
	after the dual update in line~\ref{line:yra2} all edges incident on $a$ become tight.
	Either $a$ gets matched in line~\ref{line:lineif}
	or the loop in line~\ref{line:loop2a2} executes at least once.
	Since $\BBB(a) \neq \emptyset$ every time the loop at line~\ref{line:loop2a2}
	is entered, an agent $a'$ is selected in line~\ref{line:rprimea2}.
	By the choice of $a'$, Lemma~\ref{lem:c1_unm_cond}, {\bf Fact 2}
	and the choice of $p$ in line~\ref{line:pickp},
	the dual update in line~\ref{line:za2} ensures that at least one edge $(a',p_j)$, for $p_j \mpref_{a'} M(a')$
	becomes matchable and $a'$ gets matched along that edge.
	By {\bf (P2)}, this modification does not demote $a'$ (if $a'$ was already matched).
	Therefore, either an unmatched agent (either $a$ in line~\ref{line:lineif}
	or $a'$ in line~\ref{line:pro1a2}) 
	is matched or at least one agent ($a'$ in line~\ref{line:pro1a2}) is promoted during an iteration.

	Thus after $O(m)$ iterations of the loop in line~\ref{line:loop_a},
	a fixed unmatched agent $a$ gets matched and the loop in line~\ref{line:loop_a} terminates.
	As mentioned earlier, the free-promotions routine takes $O(m^2)$ time.
	Thus, the loop in line~\ref{line:loop_a} terminates in $O(m^3)$ time for a fixed unmatched agent $a$
	and the loop in~\ref{line:loop1a2} terminates in $O(m^3$$\mid\hspace{-0.1cm}\AAA\hspace{-0.1cm}\mid)$ time. 
	By the termination condition of the loop, $M$ is an $\AAA$-perfect matching.
\end{proof}
	
\noindent{\bf Remark on the running time. }
We observe that the initial setting of dual variables takes 
$O(m$$\mid\hspace{-0.1cm}\AAA\hspace{-0.1cm}\mid)$ time because there are $O(m$$\mid\hspace{-0.1cm}\AAA\hspace{-0.1cm}\mid)$ valid triplets.
Since the algorithm guarantees {\bf (P2)}, with a careful implementation of free-promotions routine
and efficiently computing the threshold agents, the running time of algorithm can be improved.
\qed

Finally, we show that the matching $M$ computed by Algorithm~\ref{algo:dualalgo_2}
is an $\ell_a$-approximation (Lemma~\ref{lem:costana}). 
\begin{lemma}\label{lem:costana}
Matching $M$ computed by Algorithm~\ref{algo:dualalgo_2} is an $\ell_a$-approximation of $\SMFQUSET$.
\end{lemma}
\begin{proof}
	Let $\OPT$ be an optimal matching and $c(M)$ and $c(\OPT)$ denote the cost of $M$ and $\OPT$ respectively.
	By the LP duality, $c(\OPT) \geq \sum\limits_{a\in\AAA}{y_a}$.
	By {\bf (P4)}, $(a,p) \in M$ implies that the edge $(a,p)$ is tight.
	Thus, we have
\begin{align*}
c(M) = \sum\limits_{(a,p) \in M}{c(p)}
	&= \sum\limits_{(a,p) \in M} \Big(y_a + \sum_{\substack{p' = p\ \text{or}\\ p'\lpref_{a} p}}\ \ {\sum\limits_{a' \lpref_{p'} a}{z_{a,p',a'}}} - \sum\limits_{a \lpref_p a'}{z_{a',p,a}}\Big)\\
	&= \sum\limits_{a\in\AAA}{y_a} + \underbrace{\sum\limits_{(a,p) \in M} \Big(\sum_{\substack{p' = p\ \text{or}\\ p' \lpref_{a} p}}\ \ {\sum\limits_{a' \lpref_{p'} a}{z_{a,p',a'}}} - \sum\limits_{a\lpref_p a'}{z_{a',p,a}}\Big)}_{S(Z)}
\end{align*}
	where the first equality is from Eq.~\ref{eq:lp1}, the second equality is from Eq.~\ref{eq:lp5} and the third equality follows because $M$ is $\AAA$-perfect.
	Let $S(Z)$ denote the second summation in the above cost.
	Our goal is to show that $S(Z)$ is upper-bounded by $(\ell_a-1)\sum\limits_{a\in\AAA}{y_a}$ thereby implying that $c(M) \leq \ell_a\cdot\sum\limits_{a\in\AAA}{y_a}$.

We first note that all the $z$ variables are set to $0$ initially
	and they are updated only inside the loop at line~\ref{line:loop2a2}.
	We charge the update in every $z$ variable to a specific unmatched agent $a$ picked at line~\ref{line:picka} and upper-bound the total update in $z$ charged to $a$ in terms of $y_a$.
	Let $A'$ be the set of agents unmatched before the loop at line~\ref{line:loop1a2} is entered.
    During every iteration of the loop in line~\ref{line:loop1a2},
	an unmatched agent $a$ from $A'$ is picked and the loop in line~\ref{line:loop_a} executes until $a$ is matched.
	Suppose that after picking $a$ in line~\ref{line:picka}, the loop in line~\ref{line:loop_a} runs for $\kappa(a)$ iterations.
	Then, $y_a$ is incremented by $c_2-c_1$ for $\kappa(a)$ times and since $a$ is matched,
	it is not picked again at line~\ref{line:picka}.
	Thus, at the end of algorithm, $y_a = c_1 + \kappa(a) (c_2-c_1)$, that is $y_a \geq \kappa(a) (c_2-c_1)$.

	We first present a simpler analysis that proves an $(\ell_a+1)$-approximation.
	Recall that the $z$ variables are non-negative (Eq.~\ref{eq:lp7}).
	Thus, we upper-bound the total value of $z$ variables appearing in positive form in $S(Z)$.
	During the iterations $1$ to $\kappa(a)-1$, the algorithm must enter the {\bf else} part
	and in the $\kappa(a)\text{-}th$ iteration, the loop may or may not enter the {\bf else} part.
	Suppose the algorithm enters the {\bf else} part. Then by Claim~\ref{cl:boundp},
	for a fixed $a$ when the algorithm enters the {\bf else} part,
	at most $\ell_a$ many $z$ variables are set to $c_2-c_1$.
	Thus, at most $\kappa(a) \ell_a (c_2-c_1)$ total update in $S(Z)$ occurs during execution
	of the loop in line~\ref{line:loop_a} when agent $a$ is picked.
	We charge this cost to agent $a$, thus agent $a \in A'$ is charged at most $\ell_a y_a$.
	Thus,
\begin{align*}
	c(M) = \sum\limits_{a\in\AAA}{y_a} + S(Z)
	&\leq \sum\limits_{a\in\AAA\setminus A'}{y_a} + \sum\limits_{a\in A'}{y_a} + \sum\limits_{a\in A'}{\ell_a y_a}\\
	&\leq (\ell_a+1) \sum\limits_{a\in\AAA}{y_a}
	\leq (\ell_a+1) c(\OPT)
\end{align*}

    Now, we proceed to a better analysis that shows an $\ell_a$-approximation.
	Recall that if $(a',p,a)$ is a valid triplet then the variable $z_{a',p,a}$ appears in the dual constraint
    of possibly multiple edges incident on $a'$ in positive form and in the dual constraint of exactly one edge,
	that is, the edge $(a,p)$ in negative form. We show that there exist certain valid triplets such that
	the corresponding $z$ variable occurring in positive form in the dual constraint of a matched edge 
	also appears in negative form in the dual constraint of another matched edge, thereby canceling out 
	their contribution in $S(Z)$. Thus, it is enough to upper-bound the update in $z$ variables that
	are {\em not} cancelled. We prove that the total update in such
    $z$ variables that is charged to an agent $a \in A'$ can be upper-bounded by $(\ell_a-1) y_a$
    instead of $\ell_a y_a$ as done earlier.

	Let $a \in A'$ be an arbitrary agent.
	Suppose that after $a$ is selected at line~\ref{line:picka},
	$a$ is matched to some program $\overline{p}$ and that $M(a) = p_k$ at the end of the algorithm.
	By {\bf (P2)}, $p_k = \overline{p}$ or $p_k \mpref_{a} \overline{p}$.
	Also, during iterations $1$ to $\kappa(a)-1$,
	$thresh(p_k) \neq a$ and the loop in line~\ref{line:loop2a2} executes.
	It implies that in each of the iterations, there exists an agent $a_j$
	such that $thresh(p_k) = a_j$ and $z_{a_j,p_k,a}$ is updated.
	Also, $a_j$ was matched to $p'$ such that $p' = p_k$ or $p' \mpref_{a_j} p_k$.
	By {\bf (P2)}, at the end of the algorithm, $M(a_j) = p'$ or $M(a_j) \mpref_{a'} p'$.
	Thus, the variable $z_{a_j,p_k,a}$ appears in positive form in the dual constraint of 
	the edge $(a_j,M(a_j))$.
	Since $(a,p_k) \in M$ and the variable $z_{a_j,p_k,a}$ appears in negative form in the 
	dual constraint of edge $(a,p_k)$.
	Therefore, the variable $z_{a_j,p_k,a}$ cancels out in $S(Z)$.
	This implies that for each of the iterations $1$ to $\kappa(a)-1$,
	at most $\ell_a-1$ many $z$ variables are set to $c_2-c_1$ such that they may not cancel out.
	We charge the update in these variables to $a$.

	In the last $\kappa(a)$-th iteration, $a$ gets matched.
	If $a$ is matched at line~\ref{line:lineif} then no $z$ variable is updated during this iteration.
	Otherwise, $a$ is matched in one of the iterations of the loop in line~\ref{line:loop2a2}
	by the free-promotions routine.
	Recall that by our assumption, $a$ is matched to $\overline{p}$ in this step.
	By the choice of $\overline{p}$ in the free-promotions routine,
	the edge $(a,\overline{p})$ must have been matchable, that is, it is tight and $thresh(\overline{p}) = a$.
	The fact that edge $(a,\overline{p})$ was tight implies (by {\bf Fact 2}) that no variable
	of the form $z_{\times,\overline{p},a}$ was updated so far inside the loop at line~\ref{line:loop2a2}
	during the $\kappa(a)$-th iteration.
	When $\BBB(a)$ is re-computed, $\overline{p} \notin \BBB(a)$ because $M(a) = \overline{p}$ at this step.
	Thus, in the subsequent iterations of the loop in line~\ref{line:loop2a2},
	no agent $a'$ could have selected $\overline{p}$ in line~\ref{line:pickp}.
	This implies that no $z$ variable of the form $z_{\times,\overline{p},a}$ is updated
	during the rest of the execution of the loop at line~\ref{line:loop2a2} of the $\kappa(a)$-th iteration.
	This implies that during the $\kappa(a)$-th iteration,
	the $z$ variables that are set to $c_2-c_1$ are of the form
	$z_{\times,\hat{p},a}$ where 
	$\overline{p} \neq \hat{p}$. By the fact that $\overline{p} \in \mathcal{N}(a)$, $\hat{p} \in \mathcal{N}(a)$
	and $\mid\hspace{-0.1cm}\mathcal{N}(a)\hspace{-0.1cm}\mid$ $\leq \ell_a$, the number such $z$ variables is at most $\ell_a-1$.

	Thus, during $\kappa(a)$ many iterations for the agent $a \in A'$ at most $\kappa(a) (\ell_a-1) (c_2-c_1)$ total update in $S(Z)$ is charged to $a$.
	Recall that $y_a \geq \kappa(a) (c_2-c_1)$.
	Thus, agent $a \in A'$ contributes at most $(\ell_a-1) y_a$ in $S(Z)$.
This gives
\begin{align*}
	c(M) = \sum\limits_{a\in\AAA}{y_a} + S(Z)
	&\leq \sum\limits_{a\in\AAA\setminus A'}{y_a} + \sum\limits_{a\in A'}{y_a} + \sum\limits_{a\in A'}{(\ell_a-1)\cdot y_a}\\
	&\leq \ell_a \sum\limits_{a\in\AAA}{y_a} \leq \ell_a \cdot c(\OPT)
\end{align*}
	This completes the proof of the lemma.
\end{proof}

This establishes Theorem~\ref{thm:c1c2}.
Finally, we show that the analysis of our algorithm is tight by presenting
a family of $\CCQCC$ instances for a fixed $\ell_a$ in Fig.~\ref{fig:ella_tight}.
Let $k \geq 1$.
We have $\ell_a k+1$ agents and $\ell_a+1$ programs.
Program $p_0$ has cost 0 and rest of the programs have a unit cost.
Agent $a_0$ ranks program $p_1, \ldots, p_{\ell_a}$ is an arbitrary order.
Rest of the agents are partitioned into $\ell_a$ groups of $k$ each.
Each agent $a_{u,v}$ where $1 \leq u \leq \ell_a$ and $1 \leq v \leq k$
ranks program $p_u$ followed by program $p_0$.
There exists an arbitrary ordering on agents such that $a_0$ is the least-preferred agent.
Note that in an optimal solution $a_0$ will be matched to some program $p_t$ in its preference list
such that the $k$ agents of the form $a_{t,w}, 1 \leq w \leq k$
must be matched to $p_t$ and the rest of the agents are matched to $p_0$.
Therefore, the cost of an optimal solution is $k+1$.
Algorithm~\ref{algo:dualalgo_2} begins with matching every $a_{u,v}$ to $p_0$
and $a_0$ is unmatched. The algorithm picks $a= a_0$ in line~\ref{line:picka}.
In every iteration of the loop at line~\ref{line:loop_a},
all the edges incident on $a_0$ become tight, but none is matchable.
The algorithm will promote the threshold at every program $p_t \neq p_0$.
This continues for $k$ iteration of the loop at line~\ref{line:loop_a} resulting in promoting
$\ell_a k$ agents ($a_{u,v}, 1 \leq u \leq \ell_a, 1 \leq v \leq k$) to the respective program $p_u$.
Then the algorithm finally matches $a_0$ and terminates.
Thus, the algorithm computes a matching with cost $\ell_a k +1$. As $k$ increases, the approximation
guarantee reaches $\ell_a$.

\begin{figure}[!ht]
	\begin{minipage}{0.5\textwidth}
		\begin{align*}
			a_{u,v} &: p_u \mpref p_0\\
			a_0 &: p_1 \mpref \ldots \mpref p_{\ell_a}
		\end{align*}
	\end{minipage}
	\begin{minipage}{0.5\textwidth}
		\begin{align*}
			(0)\ p_0 &: a_{1,1} \mpref \ldots a_{\ell_a,k}\\
			(1)\ p_u &: a_{u,1} \mpref \ldots a_{u,k} \mpref a_0
		\end{align*}
	\end{minipage}
	\vspace{0.3cm}
	\caption{A family of instances that illustrate the tightness of analysis of Algorithm~\ref{algo:dualalgo_2}.
	Here, $1 \leq u \leq \ell_a$ and $1 \leq v \leq k$.
	}
	\label{fig:ella_tight}
\end{figure}

The restriction of two distinct costs is crucially used in the 
analysis of our approximation algorithm for $\SMFQUSET$ on $\CCQCC$ instances.
An interesting open question is to use the LP for general instances.

\section{$\SMFQUSET$ on general instances and $\minmax$}\label{sec:ucalgo}
In this section, we present our algorithmic results for the $\SMFQUSET$ problem on general $\SMFQ$ instances
and for the $\minmax$ problem.
\subsection{$\ell_p$-approximation for $\SMFQUSET$}\label{sec:lpapprox}
Let $H$ be a $\SMFQ$ instance and let $p^*_a$ denote the minimum cost program in the preference list of agent $a$. 
If there is more than one program with the same minimum cost, we let $p^*_a$ be the most-preferred such program.

\noindent  {\bf Description of the first algorithm: } 
Our algorithm (Algorithm~\ref{algo:algo_mincostset_arb}) 
starts by matching every agent $a$ to $p^*_a$. Note that such 
a matching is $\AAA$-perfect and min-cost but not necessarily envy-free. 
Now the algorithm considers programs in an arbitrary order. 
For program $\cc$, we consider agents in the reverse preference list ordering of $\cc$.
If there exists agent $a \notin M(\cc)$ such that $\cc \mpref_a M(a)$ and there exists $a' \in M(\cc)$ such that $a' \lpref_{\cc} a$, then $(a,a')$ form an envy-pair. We resolve this
by promoting $a$ from $M(a)$ to $p$.
The algorithm stops after considering every program. 

\begin{algorithm}
	\begin{algorithmic}[1]
		\State let $M = \{(\s, \cc)\ \mid \ a \in \AAA$ and $ p = p^*_a\}$\label{line:beforeloop1}
		\For {every program ${\cc}$ }\label{line:loop}
		\For {$\s$ in reverse preference list ordering of $\cc$}\label{line:a}
			\If {there exists  $\s' \in M(\cc)$ such that $a \mpref_{p} a'$ and $p \mpref_a M(a)$} 
				\State $M = M \setminus \{(\s, M(\s))\} \cup \{(\s, \cc)\}$\label{line:promote}
			\EndIf
		\EndFor
		\EndFor
		\State return $M$
	\end{algorithmic}
	\caption{An $\ell_p$-approximation algorithm for $\SMFQUSET$}
	\label{algo:algo_mincostset_arb}
\end{algorithm}

Note that in Algorithm~\ref{algo:algo_mincostset_arb}
an agent may get promoted in the loop (line~\ref{line:promote})
of the algorithm but never gets demoted.
Further, program $p$ is assigned agents in the {\bf for} loop (line~\ref{line:loop})
only when at least one agent is matched to $p$ in line~\ref{line:beforeloop1}.
Therefore, if program $p$ is assigned at least one agent
in the final output matching, then $p = p^*_a$ for some agent $a \in \AAA$.

\noindent {\bf Analysis:} 
It is clear that the matching computed by Algorithm~\ref{algo:algo_mincostset_arb} is $\AAA$-perfect.
Now we show that output of the algorithm is an envy-free matching
and an $\ell_p$-approximation to an optimal solution of the $\SMFQUSET$ problem.

\begin{lemma}\label{lem:corr2}
	The matching $M$ output by Algorithm~\ref{algo:algo_mincostset_arb} is envy-free. 
\end{lemma}
\begin{proof}
	We show that no agent participates in an envy-pair w.r.t. $M$.
	Assume for contradiction, that $(\s, \s')$ form an envy-pair w.r.t. $M$. 
	Suppose $p = M(a')$.
	Then $\cc \mpref_{\s} M(\s)$ and $\s' \lpref_{\cc} \s$. 
	Consider the iteration of the {\bf for} loop in line~\ref{line:loop}
	when $\cc$ was considered.
	Agent $a'$ was either already matched to $p$ (before the iteration began)
	or is assigned to $p$ in this iteration.
	Note that $a \mpref_p a'$ and the agents are considered in reverse order of $p$'s preferences.
	Thus, in either case when $a$ was considered in line~\ref{line:a}, $a' \in M(p)$.
	If $M(a) = p $ or $M(a) \mpref_a p$ at this line then $M(a) = p $ or $M(a) \mpref_a p$ at the end of the algorithm,
	since an agent never gets demoted.
	Thus, we must have that $M(a) \lpref_a p$ at this line.
	This implies that the algorithm matched $a$ to $p$.
	Since $a$ can only get promoted during the subsequent iterations, $M(a) = p$ or $M(a) \mpref_a p$
	at the end of the algorithm. This contradicts the claimed envy.
\end{proof}

Next we show the approximation guarantee using a natural lower-bound on the optimal cost of $\SMFQUSET$.
Let $\OPT$ denote an optimal solution and $c(\OPT)$ denote the cost of $\OPT$.
Since $\OPT$ is $\AAA$-perfect, it is easy to observe that $c(\OPT) \geq \sum\limits_{\s \in \AAA}{c(p_{\s}^*)}$.
We denote this lower-bound as $lb_1$. Using $lb_1$ we show that the output matching $M$ is an $\ell_p$-approximation.

\begin{lemma}\label{lem:lp_app2}
The matching $M$ output by Algorithm~\ref{algo:algo_mincostset_arb} is an $\ell_p$-approximation.
\end{lemma}
\begin{proof}
	Let $c(M)$ denote the cost of matching $M$.
In the matching $M$, some agents are matched to their least-cost program (call them $\AAA_1$), whereas some agents get promoted (call them
$\AAA_2$).  
	Recall that if a program $p$ is assigned agents in $M$ then 
	$p = p^*_a$ for some agent $a$.
	Thus for an agent $a \in \AAA_2$, 
	we charge the cost of some other least-cost program $p^*_{a'}$ such that $a' \neq a$.
	A program can be charged at most $\ell_p -1$ times for the agents in $\AAA_2$, thus
\begin{equation*}
	c(M) = \sum\limits_{\s \in \AAA_1}{c(p^*_a)} + \sum\limits_{\s \in \AAA_2}{c(M(a))} \le \sum\limits_{\s \in \AAA}{c(p^*_a)} + \sum\limits_{\s \in \AAA}{(\ell_p -1)\cdot c(p^*_a)} \le \ell_p \cdot c(\OPT)
\end{equation*}
This completes the proof of the lemma.
\end{proof}

Next, we present another algorithm 
with approximation guarantee $\ell_p$.

\noindent  {\bf Description of the second algorithm ({\sf ALG}): } Given a $\SMFQ$ instance $H$, we construct a subset $\BBB'$ of
$\BBB$ such that $p \in \BBB'$ iff $p = p^*_a$ for some agent $a$. Our algorithm now matches every agent $a$ to the most-preferred program in $\BBB'$. 

\noindent {\bf Analysis of {\sf ALG}:}
It is clear that the matching computed by {\sf ALG} is $\AAA$-perfect.
Let $M$ be the output of {\sf ALG} and $\OPT$ be an optimal matching. 
Let $c(\OPT)$ and $c(M)$ be the cost of matching $\OPT$ and $M$ respectively.
The lower-bound $lb_1$ on $c(\OPT)$ is exactly the same.
We show the correctness and the approximation guarantee of {\sf ALG} via Lemma~\ref{lem:corr1} and 
Lemma~\ref{lem:basic_lp}.

\begin{lemma}\label{lem:corr1}
The output $M$ of {\sf ALG} is envy-free.
\end{lemma}
\begin{proof}
	We show that no agent participates in an envy-pair w.r.t. $M$.
	Suppose that there exists an envy-pair $(a,a')$
	that is, there exists program $p$ such that $\cc \mpref_{\s} M(\s)$ and there exists 
	an agent $\s' \in M(\cc)$ such that $\s' \lpref_{\cc} \s$. Since
	{\sf ALG} assigns agents to programs in $\BBB'$ only, it implies that $\cc \in \BBB'$.
	However, $M(a)$ is the most-preferred program in $\BBB'$ and hence 
	$M(\s) \mpref_{\s} \cc$ or $M(\s) = \cc$. Thus, the claimed envy-pair does not exist.
\end{proof}

\begin{lemma}\label{lem:basic_lp}
The output $M$ of {\sf ALG} is an $\ell_p$-approximation.
\end{lemma}
\begin{proof}
In the matching $M$,  agent $a$ is either matched to $p^*_a$ 
or $p^*_{a'}$ for some other agent $a'$. This is determined by the relative 
ordering of $p^*_a$ and $p^*_{a'}$ in the  preference list of $a$.
We partition the agents as $\AAA = \AAA_1 \cup \AAA_2$, where $\AAA_1$ is the set of agents matched
to their own least cost program, that is, $a \in \AAA_1$ iff $M(a) = p^*_a$. We define $\AAA_2 = \AAA \setminus \AAA_1$.
We can write the cost of $M$ as follows:
\begin{equation*}
	c(M) = \sum\limits_{\s \in \AAA_1}{c(p^*_a)} + \sum\limits_{\s \in \AAA_2}{c(M(a))}
\end{equation*}
	By similar arguments as in Lemma~\ref{lem:lp_app2}, we get the $\ell_p$-approximation guarantee.
\end{proof}

This establishes Theorem~\ref{thm:smfq_approx}.\ref{part3}.

\subsubsection{Comparing the two $\ell_p$-approximation algroithms for $\SMFQUSET$ } \label{sec:compare}
We present following instances which illustrate that neither of the two algorithms (Algorithm~\ref{algo:algo_mincostset_arb}
and $\sf ALG$)
is strictly better than the other.

\begin{example}\label{ex:ex1}
	Let $\AAA = \{a_1, a_2, \ldots, a_n\}$,
	$\BBB = \{p_1, p_2\}$, $c(p_1) = 1, c(p_2) = \alpha$ 
	where $\alpha$ is some large positive constant.
The agents $a_1, \ldots, a_{n-1}$ have the same preference list $p_2$ followed by $p_1$. 
Whereas agent $a_n$ has only $p_2$ in its preference list. The preferences of the programs are
as given below.

	\begin{minipage}{0.9\textwidth}
		\begin{align*}
			\cc_1 &: \s_1,\ \s_2,\ \ldots,\ \s_{n-1}\\
			\cc_2 &: \s_n,\ \s_{n-1},\ \s_{n-2},\ \ldots,\ \s_1
		\end{align*}
	\end{minipage}

\vspace{0.1in}
\noindent Here, {\sf ALG} outputs  $M_1$  of cost $n \cdot \alpha$ 
where  $M_1 = \{(\s_1, \cc_2), \ldots, (\s_n, \cc_2)\}$.
In contrast, Algorithm~\ref{algo:algo_mincostset_arb} outputs $M_2 = \{(\s_1, \cc_1), \ldots, (\s_{n-1},\cc_1), (\s_n,\cc_2)\}$ whose cost 
	is $n-1 +  \alpha$. Clearly, Algorithm~\ref{algo:algo_mincostset_arb} outperforms {\sf ALG} in this case and in fact  $M_2$ is optimal for the instance.

\end{example}
\begin{example}\label{ex:ex3}
Let $\AAA = \{a_1, a_2, \ldots, a_n\}$, 
	$\BBB = \{p_1, p_2, p_3\}$, $c(p_1) = 1, c(p_2) = 2, c(p_3) = \alpha$ 
	where $\alpha$ is some large positive constant. The preferences of
agents $a_1, \ldots, a_{n-2}$ are $p_2$ followed by $p_3$ followed by $p_1$. The preference list of $a_{n-1}$ contains only
$p_2$ and the preference list of $a_n$ contains only $p_3$. 
	The preferences of programs are as shown below.

	\begin{minipage}{0.9\textwidth}
		\begin{align*}
			\cc_1 &: \s_1,\ \s_2,\ \ldots, \s_{n-2}\\
			\cc_2 &: \s_{n-1},\ \s_1,\ \s_2,\ \ldots, \s_{n-2}\\
			\cc_3 &: \s_1,\ \ldots,\ \s_{n-2},\ \s_n
		\end{align*}
	\end{minipage}

\vspace{0.1in}
\noindent Here, {\sf ALG} outputs $M_1 =  \{(\s_1, \cc_2), \ldots, 
	(\s_{n-1}, \cc_2), (\s_n, \cc_3)\}$ whose
	cost is $2 \cdot (n-1) + \alpha$. In contrast, Algorithm~\ref{algo:algo_mincostset_arb} outputs $M_2$ of cost $2+ (n-1) \cdot \alpha$ where
$M_2 = \{(\s_1, \cc_3), \ldots, (\s_{n-2},\cc_3), (\s_{n-1}, \cc_2), (\s_n, \cc_3)\}$. In this instance
	{\sf ALG} outperforms Algorithm~\ref{algo:algo_mincostset_arb} and it can be verified that $M_1$ is the optimal matching.
\end{example}

\subsubsection{Discussion on $lb_1$}
Although $lb_1$ is a natural lower-bound, we show that the best approximation guarantee
using $lb_1$ is $\ell_p$, thereby showing that the analysis of our algorithms is tight.
In Fig.~\ref{fig:weak_lb}, 
we present a family of instances with $n$ agents and $3$
programs where $c(p_0) = 0$, $c(p_1) = 1$, $c(p_2) = n$ and $\ell_p=n$.
The lower-bound $lb_1 = 1$ and an optimal matching in this instance has cost $n = \ell_p \cdot lb_1$.
We remark that this holds even under master list ordering on agents and programs.
\begin{figure}[!ht]
	\begin{minipage}{0.5\textwidth}
		\begin{align*}
    			1 \leq i \leq n-1, \ \ \ \s_i &: \cc_1 \mpref \cc_0\\
			\s_n &: \cc_1 \mpref \cc_2
		\end{align*}
	\end{minipage}\hfill
	\begin{minipage}{0.5\textwidth}
		\begin{align*}
			(0)\ \cc_0 &: \s_1 \mpref \s_2 \mpref \ldots \mpref \s_{n-1}\\
			(1)\ \cc_1 &: \s_1 \mpref \s_2 \mpref \ldots \mpref \s_{n-1} \mpref \s_n\\
			(n)\ \cc_2 &: \s_n
		\end{align*}
	\end{minipage}
	\vspace{0.3cm}
	\caption{A family of instances with optimal cost exactly $\ell_p$ times the lower-bound $lb_1$. There are two optimal matchings of cost $n$:
	$\OPT_1 = \{(a_1,p_0), \ldots, (a_{n-1},p_0), (a_n,p_2)\}$ and $\OPT_2 = \{(a_1,p_1), \ldots, (a_n,p_1)\}$.}
	\label{fig:weak_lb}
\end{figure}

\subsection{$\minmax$ problem and a $\mid\hspace{-0.1cm}\BBB\hspace{-0.1cm}\mid$-approximation for $\SMFQUSET$ }\label{sec:papprox}
Let $H$ be a $\SMFQ$ instance.
Our algorithm for the $\minmax$ problem is based on the following observations:
let $M^*$ be an optimal solution for the $\minmax$ problem on $H$. 
\begin{itemize}
	\item 	Let $t^* = \max\limits_{p \in \BBB}\ \{c(p)\cdot \mid\hspace{-0.1cm}M^*(p)\hspace{-0.1cm}\mid\}$.
		Consider an $\HR$ instance $G_{t^*}$ where the preference lists are borrowed from $H$ and 
		for each $p \in \BBB$, $q(p) = \left \lfloor \frac{t^*} {c(p)} \right \rfloor$.
		Note that $M^*$ is an $\AAA$-perfect envy-free matching in $H$ such that
		$\mid\hspace{-0.1cm}M^*(p)\hspace{-0.1cm}\mid \leq q(p)$ for every $p \in \BBB$.
		Thus the instance $G_{t^*}$ admits an $\AAA$-perfect stable matching.
	\item 
		For any $t < t^*$, consider an $\HR$ instance $G_t$
		with $q(p) = \left \lfloor \frac{t} {c(p)} \right \rfloor$. The instance $G_t$ does 
		not admit an $\AAA$-perfect stable matching. 
		Otherwise this contradicts the optimality of $M^*$. 
	\item The optimal value $t^*$ lies in the range $\min\limits_{p\in\BBB}\ \{c(p)\}$ to $\max\limits_{p\in\BBB}\ \{(len(p) \cdot c(p))\}$, where $len(p)$ denotes the length of the preference list of program $p$.
\item For any $t' > t^*$, $G_{t'}$ admits an $\AAA$-perfect envy-free matching.
\end{itemize}

Our algorithm begins by constructing a sorted array $\hat{c}_p$ for each program $p \in \BBB$
such that for $1 \leq i \leq len(p)$, we have $\hat{c}_p[i] = i \cdot c(p)$.
There are $\mid\hspace{-0.1cm}\BBB\hspace{-0.1cm}\mid$ many such sorted arrays and the total number of elements in these arrays is $\sum\limits_{p \in \BBB}{len(p)} = \mid\hspace{-0.1cm}E\hspace{-0.1cm}\mid = m$.
We merge these arrays to construct a sorted array $\hat{c}$ of distinct costs.
Then we perform a binary search for the optimal value of $t^*$ in the sorted array $\hat{c}$:
for a particular value $t = \hat{c}[k]$ we construct the $\HR$
instance $G_t$ by setting appropriate quotas.
If the stable matching in $G_t$ is not $\AAA$-perfect, then
we search in the {\em upper-range}.
Otherwise, we check if $G_{t'}$ admits
an $\AAA$-perfect stable matching for $t' = \hat{c}[k-1]$. If not, we return $t$ otherwise,
we search for the optimal in the {\em lower-range}. 

\noindent{\bf Running time. }
The sorted array $\hat{c}$ is constructed by merging $\mid\hspace{-0.1cm}\BBB\hspace{-0.1cm}\mid$ sorted arrays containing $m$ total entries.
Thus, this step takes $O(m \log{\mid\hspace{-0.1cm}\BBB\hspace{-0.1cm}\mid})$ time.
The algorithm requires $O(\log{m})$ iterations because there are at most $m$ distinct costs in $\hat{c}$.
Each iteration computes at most two stable matchings using the linear time algorithm~\cite{GS62}. 
Thus the overall running time is $O(m \log {m})$.

This establishes Theorem~\ref{thm:minmaxInP}.
We now prove that the optimal matching $M^*$ for the $\minmax$ problem is a $\mid\hspace{-0.1cm}\BBB\hspace{-0.1cm}\mid$-approximation for the $\SMFQUSET$ problem.
\begin{lemma}
The optimal solution for the $\minmax$ problem is a $\mid\hspace{-0.1cm}\BBB\hspace{-0.1cm}\mid$-approximation for the $\SMFQUSET$ problem.
\end{lemma}
\begin{proof}
Let $H$ be a $\SMFQ$ instance and let $M^*$ be the optimal matching for the $\minmax$ problem on $H$. For the same instance $H$,
let $N^*$ be the optimal matching for the $\SMFQUSET$ problem. Let us define $t^*$ and $y^*$ as follows:
\begin{eqnarray*}
	t^* = \max_{p \in \BBB} \{c(p) \cdot \mid\hspace{-0.1cm}M^*(p)\hspace{-0.1cm}\mid \}  \hspace{1.8in}
	y^* = \sum_{p \in \BBB} \left (c(p) \cdot \mid\hspace{-0.1cm}N^*(p)\hspace{-0.1cm}\mid \right) 
\end{eqnarray*}
We first observe that $y^* \ge t^*$. This is true because $N^*$ is an $\AAA$-perfect envy-free matching in $H$. Furthermore, 
	since costs of all programs are non-negative, $\max_{p \in \BBB} \{c(p) \cdot \mid\hspace{-0.1cm}N^*(p)\hspace{-0.1cm}\mid \} \le y^*$. Therefore if $y^* < t^*$,
it contradicts the optimality of $M^*$ for the $\minmax$ problem.
To prove the approximation guarantee, we note that since $t^*$ denotes the maximum cost incurred at any program in $M^*$, the total cost of $M^*$
is upper bounded by $\mid\hspace{-0.1cm}\BBB\hspace{-0.1cm}\mid$$\cdot  t^* \le $$\mid\hspace{-0.1cm}\BBB\hspace{-0.1cm}\mid$$\cdot   y^*$. 
\end{proof}
This establishes Theorem~\ref{thm:smfq_approx}.\ref{part2}.

\section{Hardness results}\label{sec:uchardness}

In this section, we prove our $\NP$-hardness and inapproximability result
for the $\SMFQUSET$ problem
(Theorem~\ref{thm:smfq_hardness}).

\subsection{$\NP$-hardness of $\SMFQUSET$}\label{apsec:hardness}
We show that the $\SMFQUSET$ problem is $\NP$-hard even under severe restrictions on the instance. 
In particular we show that the hardness holds even when all agents have a preference list
of a constant length $f \geq 2$ 
and there is a master list ordering on agents and programs.
To show the $\NP$-hardness of $\SMFQUSET$ we use the Set Cover ($\SC$) instance where every element occurs in exactly $f$ sets.
Minimum vertex cover on $f$-uniform hypergraphs is known to be $\NP$-complete and
$\SC$ problem where every element occurs in exactly $f$ sets is equivalent to it~\cite{CARDINAL201267}.

\noindent {\bf Reduction:}
Let $\langle S,E,k \rangle$ be an instance of $\SC$ such that every element in $E$ occurs
in exactly $f$ sets.
Let $m = \mid\hspace{-0.1cm}S\hspace{-0.1cm}\mid, n = \mid\hspace{-0.1cm}E\hspace{-0.1cm}\mid$.
We construct an instance $H$ of $\SMFQ$ as follows.
For every set $s_i \in S$, we have a set-agent $\s_i$ and a set-program $\cc_i$.
For every element $e_h \in E$, we have an element-agent $\s_h'$.
We also have a program $p$ and $f-2$ programs $\cc^1, \ldots, \cc^{f-2}$.
Thus in the instance $H$ we have $m+n$ agents and $m+f-1$ programs.
Let $E_i$ denote the set of elements in the $s_i$. 
The element-agents corresponding to $E_i$ in $H$ are denoted by $\AAA_i'$.

\noindent {\bf Preferences:} 
The preference lists 
are shown in Fig.~\ref{fig:mincostset_hardness_sc}.
Every set-agent $\s_i$ has $f$ programs in its preference list --
the set-program $\cc_i$ followed by the program $\cc$ followed by the programs $\cc^1, \ldots, \cc^{f-2}$
in that order. 
Every element-agent $\s'_j$ has the $f$ set-programs corresponding to the sets that contain it in an arbitrary fixed order.
Every set-program $\cc_i$ has its set-agent $\s_i$ as its top-preferred agent
followed by the agents in $\AAA_i'$ in an arbitrary fixed order. The program $\cc$ 
and each program $\cc^1, \ldots, \cc^{f-2}$ has 
the set-agents as $\s_1, \ldots, \s_m$ in an arbitrary fixed order. 
It is clear that every agent has a preference list of length $f$.

\noindent{\bf Costs:} The costs for program $p_i$ for $1 \leq i \leq m$ and $p^j$ for $1 \leq j \leq f-2$ is $1$ and that for the program $\cc$ is $0$.
Thus, the instance has two distinct costs.

\begin{figure}[!ht]
        \begin{minipage}{0.45\textwidth}
	\begin{align*}
		\s_i &: \cc_i \mpref \cc \mpref \cc^1 \mpref \ldots \mpref \cc^{f-2}\\
		\s_j' &: \cc_{j1}\mpref \ldots\mpref \cc_{jf} 
	\end{align*}
	\end{minipage}\hfill
        \begin{minipage}{0.45\textwidth}
	\begin{align*}
		\cc_i &: \s_i\mpref \AAA_i'\\
		\cc &: \s_1\mpref \ldots\mpref \s_m\\
		\cc^t &: \s_1\mpref \ldots\mpref \s_m
	\end{align*}
	\end{minipage}%
	\vspace{0.3cm}
    \caption{Preference lists in the reduced instance $H$ of $\SMFQ$ from instance $\langle S,E,k \rangle$ of $\SC$.
	Here, $1 \leq i \leq m,	1 \leq j \leq n$ and $1\leq t \leq f\text{-}2$}
    \label{fig:mincostset_hardness_sc}
\end{figure}

\begin{claim}\label{cl:smfquset_cl1}
	If $\langle S,E,k \rangle$ is a yes instance, then $H$ admits
an $\AAA$-perfect envy-free matching of cost at most $n+k$.
\end{claim}
\begin{proof}
Let $X \subseteq S$ be the set cover of size at most $k$.
	Using $X$, we construct an $\AAA$-perfect matching $M$ in $H$ and show that $M$ is envy-free. 
We then show that the cost of $M$ is bounded.
For every set $s_i \notin X$, we match the corresponding set-agent $\s_i$ to the program $\cc$, that is, $M(\s_i) = p$.
For every set $s_i \in X$, we match the corresponding set-agent $\s_i$ to the program $\cc_i$, that is, $M(\s_i) = \cc_i$.
Since $X$ is a set cover, for every element $e_j \in E$, at least one of the sets it occurs in is in $X$.
Thus, for every element-agent $\s'_j$, we match it to the program $\cc_{jt}$ corresponding to the set in $X$ (in case more than one set containing $e_j$ 
are in $X$, then we match it to the program which is highest-preferred among them).
It is clear that $M$ is $\AAA$-perfect. 

To prove that $M$ is envy-free, we show that no agent participates in an envy-pair. First observe that,
an agent $\s_i$ corresponding to $s_i \in X$ is matched to its top-choice program. Hence such agents do not participate in envy pairs.
Now, for all agents $\s_i$ such that $s_i \notin X$, these agents are matched to $\cc$. However, the corresponding program $p_i$ remains
	closed.  Finally, every element-agent that is not matched to its top-choice programs has all such top-choice programs closed.
Thus the matching is envy-free.

	We compute the cost of matching $M$ from the agent side. Each element-agent is matched to some program $\cc_i$ and costs $1$ each. At most $k$ set-agents matched to their corresponding set-program $\cc_i$ each costs $1$ and at least $m-k$ set-agents are matched to program $\cc$ that incur the cost $0$ each. Hence the cost of the matching is at most $n+k$.
\end{proof}

\begin{claim}\label{cl:smfquset_cl2}
	If $H$ admits an $\AAA$-perfect envy-free matching with
cost at most $n+k$, then $\langle S,E,k \rangle$ is a yes instance.
\end{claim}
\begin{proof}
	Let $M$ be an $\AAA$-perfect envy-free matching in $H$ with cost at most $n+k$.
First we prove that program $p$ must take at least $m-k$ set-agents --
since the matching is $\AAA$-perfect,
	every matched element-agent contributes a cost of at least $1$ each, every set-agent not matched to $\cc$
	contributes a cost of $1$ each thus if program $\cc$ takes less than $m-k$ set-agents then the cost of any such $\AAA$-perfect envy-free matching
	is at least $n+k+1$. This leads to a contradiction.
Thus, at least $m-k$ set-agents must be matched to program $\cc$.

Let $X$ be the set of sets $s_i$ such that $\s_i$ is matched to $\cc_i$.
	Then, $\mid\hspace{-0.1cm}X\hspace{-0.1cm}\mid$ $\leq k$.
Since, the matching is envy-free,
every program $\cc_i$ such that
	$s_i \notin X$ must be closed  (since $\s_i$ is either matched to $p$ or $p^j$).
We now prove that $X$ is a set cover.
Suppose not, then there exists at least one element, say $e_j$ such that no set containing
$e_j$ is in $X$, implying that all the programs $\cc_{jt}$ are closed.
This implies that element-agent $\s_t'$ is unmatched, that is, the matching $M$ is not $\AAA$-perfect.
	This leads to a contradiction.
	Hence $X$ must be a set cover. Since $\mid\hspace{-0.1cm}X\hspace{-0.1cm}\mid$ $\leq k$, thus $\langle S, E, k \rangle$ is a yes instance.
\end{proof}

We note that the costs in the reduced instance of $\SMFQ$ are constant, thus $\SMFQUSET$ is strongly $\NP$-hard.
We remark that the $\NP$-hardness result holds even when 
there is a master list ordering on agents and programs as follows: 
$(\s_1 \mpref \ldots \mpref \s_m \mpref \s_1' \mpref \ldots \mpref \s_n')$, 
$(\cc_1 \mpref \ldots \mpref \cc_m \mpref \cc \mpref \cc^1 \mpref \ldots \mpref \cc^{f-2})$.

\noindent This establishes Theorem~\ref{thm:smfq_hardness}.\ref{hpart1}.

\subsection{Inapproximability of $\SMFQUSET$}\label{apsec:inapprox}
In this section, we present a reduction
from the Minimum Vertex Cover ($\MVC$) problem
to prove the inapproximability result (Theorem~\ref{thm:smfq_hardness}.\ref{hpart2}) for the $\SMFQUSET$ problem.
Let $G=(V,E)$ be an instance of $\MVC$ problem. 
We construct an instance of $\SMFQ$ as follows.

\noindent{\bf Reduction: }
Let $n = \mid\hspace{-0.1cm}V\hspace{-0.1cm}\mid$ and $m = \mid\hspace{-0.1cm}E\hspace{-0.1cm}\mid$.
For every vertex $u_i$, 
we have $m$ vertex-agents $\s_i^1, \ldots \s_i^m$ and $2$ vertex-programs $\cc_i$ and $\cc_i'$.
For every edge $e_j$, we have one edge-agent $\s'_j$.
We have an additional program $p$.
Thus, we have $m+mn$ agents and $2n+1$ programs.

\noindent {\bf Preferences: } 
Let $\BBB_j'$ denote the set of vertex-programs $\cc_{j1}'$ and $\cc_{j2}'$ corresponding to the end-points $u_{j1}$ and $u_{j2}$
of edge $e_j$ 
and $\AAA'_i$ denote the set of edge-agents $\s_j'$ corresponding to the edges 
incident on vertex $u_i$. 
Each vertex-agent $\s_i^t, 1 \leq t \leq m$ has the programs $\cc_i, \cc_i', p$ in that order. 
Each edge-agent $\s_j'$ has the two programs in $\BBB_j'$ in an arbitrary but fixed order.
Each vertex-program $\cc_i$ has 
$m$ vertex-agents $\s_i^1, \ldots, \s_i^m$ in that order.
Each vertex-program $\cc_i'$ has 
$m$ vertex-agents $\s_i^1, \ldots, \s_i^m$ in that order, followed by the edge-agents $\AAA'_i$.
Program $p$ has 
the $mn$ vertex-agents in an arbitrary, fixed order.

\noindent{\bf Costs: }
The cost of each vertex-program $\cc_i$ is $3$, that of each vertex-program $\cc_i'$ is $2n$ and that of the program $p$ is $0$.
Note that $2n > 3$ since $n \geq 2$.

\begin{figure}[!ht]
        \begin{minipage}{0.4\textwidth}
	\begin{align*}
			\s_i^t &: \cc_i \mpref \cc_i' \mpref \cc \\
			\s_j'&: \BBB_j'
	\end{align*}
        \end{minipage}\hfill
        \begin{minipage}{0.45\textwidth}
                \begin{align*}
			\cc_i &: \s_i^1 \mpref \ldots \mpref \s_i^m \\
			\cc_i' &: \s_i^1 \mpref \ldots \mpref \s_i^m \mpref \AAA_i'\\
			p &: \s_1^1 \ldots \s_n^m
                \end{align*}
        \end{minipage}
	\vspace{0.3cm}
    \caption{Preference lists in the reduced instance $H$ of $\SMFQ$ from instance $G$ of $\MVC$.
	Here, $1\leq i \leq n,\ 1 \leq j \leq m$ and $1 \leq t \leq m$}
    \label{fig:mincostset_inapprox}
\end{figure}

We prove the guarantees obtained from the reduction in the following two lemmas.

\begin{lemma}\label{lem:lem1hardness}
	If an optimal vertex cover in $G$ is of size at most $(\frac{2}{3}+\epsilon)\cdot n$ then in $H$, there exists an $\AAA$-perfect envy-free matching $M$ with cost
	at most $(4+3\epsilon)mn$.
\end{lemma}
\begin{proof}
	Let $V'$ be the vertex cover of $G$ of size at most $(\frac{2}{3}+\epsilon)\cdot n$.
	If $u_i \notin V'$, we match all the $m$ agents $\s_i^t$\ to the program $p$ at cost $0$.
	If $u_i \in V'$, we match all the $m$ agents $\s_i^t$ to the program $\cc_i$.
	This contributes the cost of at most $m\cdot (\frac{2}{3}+\epsilon) \cdot n \cdot 3 = mn(2 + 3\epsilon)$.
We then prune the preference list of every edge-agent $\s_j'$ by deleting the programs $\cc_i'$ corresponding to the end-point $u_i$ such that $u_i \notin V'$.
Since $V'$ is a vertex cover, it is guaranteed that every edge is covered and hence every edge-agent $\s_j'$ has a non-empty list $\BBB_j'$ after pruning.
Every edge-agent $\s_j'$ is then matched to the top-preferred program $\cc_i'$ in the pruned list $\BBB_j'$.
This contributes a cost of $2mn$ incurred by the edge-agents.
	Thus, the total cost of this matching is at most $(2+2+3\epsilon)mn = (4+3\epsilon)mn$.
	It is easy to see that $M$ is $\AAA$-perfect and is envy-free since no agent participates in an envy-pair.
\end{proof}

\begin{lemma}\label{lem:lem2hardness}
	If the optimal vertex cover in $G$ has size greater than $(\frac{8}{9}-\epsilon)\cdot n$ then in $H$, any $\AAA$-perfect 
	envy-free matching $M$
	has cost greater than $(\frac{14}{3}-3\epsilon)mn$.
\end{lemma}
\begin{proof} 
	We prove the contra-positive i.e. if there exists an $\AAA$-perfect 
	envy-free matching
	with cost at most $mn(\frac{14}{3}-3\epsilon)$ then optimal vertex cover in $G$ has size at most $(\frac{8}{9}-\epsilon)\cdot n$.
	Given that the matching is $\AAA$-perfect, 
	all the $\s_j'$ type agents ($m$ in total) must be matched to some program in $\BBB_j'$. 
	Note that each such agent must contribute a cost of $2n$ and hence they together contribute a cost of $2mn$.

	Suppose the edge-agent $\s_j'$ is matched to a program $\cc_i'$ such that $u_i$ is one of the end-points of edge $e_j$.
	Then, $m$ $\s_i^t$ vertex-agents for the vertex $u_i$ must be matched to either $\cc_i$ or $\cc_i'$, otherwise they form an envy pair with $\s_j'$. 
Let $V'$ be the set of vertices $u_i$ such that at least one $\s_i^t$ is matched to $\cc_i$ or $\cc_i'$.
If $u_i \in V'$ is such that less than m $\s_i^t$ type agents are matched to $\cc_i$ or $\cc_i'$, then it implies that no edge-agent corresponding to the edge incident on $u_i$ could have been matched to $\cc_i'$ but is matched to $\cc_k'$ where $u_k$ is its other end-point.
Thus, we can remove such $u_i$ from $V'$. 
It is clear that $V'$ is a vertex cover of $G$ since the matching is $\AAA$-perfect.
	For every $u_i \in V'$, all $m$ $\s_i^t$ vertex-agents are matched to $\cc_i$ or $\cc_i'$. 
	Since they together contribute the cost at most $mn(\frac{8}{3}-3\epsilon)$ and each $u_i \in V'$ has $m$ copies matched, each contributing a cost of at least $3$, it implies that $\mid\hspace{-0.1cm}V'\hspace{-0.1cm}\mid$ $\leq (\frac{8}{9}-\epsilon)\cdot n$, implying that an optimal vertex cover in $G$ has size at most $(\frac{8}{9}-\epsilon)\cdot n$.
\end{proof}

Now we show the inapproximability for $\SMFQUSET$.

\begin{lemma}\label{lem:rperfect_7by6}
	The $\SMFQUSET$ problem is $\NP$-hard to approximate within a factor of $\frac{7}{6}-\delta$ for any constant $\delta > 0$, unless $\Poly = \NP$.
\end{lemma}
\begin{proof}
We use the following proposition and the results proved in Lemma~\ref{lem:lem1hardness} and Lemma~\ref{lem:lem2hardness}.

\noindent {\bf Proposition~\cite{DS02}.} For any $\epsilon > 0$ and $p < \frac{3-\sqrt{5}}{2}$,
the following holds: If there is a polynomial-time algorithm that, given a graph
$G = (V, E)$, distinguishes between the following two cases, then $\Poly = \NP$.

	\hspace{-0.2cm} $(1) \mid\hspace{-0.1cm}VC(G)\hspace{-0.1cm}\mid$ $\leq (1 - p + \epsilon)\mid\hspace{-0.1cm}V\hspace{-0.1cm}\mid$

	\hspace{-0.2cm} $(2) \mid\hspace{-0.1cm}VC(G)\hspace{-0.1cm}\mid$ $> (1 - max\{p^2, 4p^3 - 3p^4\} - \epsilon)\mid\hspace{-0.1cm}V\hspace{-0.1cm}\mid$
\qed

	By letting $p = \frac{1}{3}$ in Proposition above, we know that the existence of
a polynomial-time algorithm that distinguishes between the following two cases
implies $\Poly = \NP$ for an arbitrary small positive constant $\epsilon$:

	\hspace{-0.2cm} $(1) \mid\hspace{-0.1cm}VC(G)\hspace{-0.1cm}\mid$ $\leq (\frac{2}{3} + \epsilon)\mid\hspace{-0.1cm}V\hspace{-0.1cm}\mid$ i.e. $cost(M) \leq (4+3\epsilon)mn$

	\hspace{-0.2cm} $(2) \mid\hspace{-0.1cm}VC(G)\hspace{-0.1cm}\mid$ $> (\frac{8}{9} - \epsilon)\mid\hspace{-0.1cm}V\hspace{-0.1cm}\mid$ i.e. $cost(M) > (\frac{14}{3}-3\epsilon)mn$

Now, suppose that there is a polynomial-time approximation algorithm $A$ for $\SMFQUSET$
whose approximation ratio is at most $\frac{7}{6}-\delta$ for some $\delta$. 
Then, consider 
	a fixed constant $\epsilon$ such that $\epsilon < \frac{8\delta}{13-6\delta}$.
	If an instance of case $(1)$ is given to $A$, it outputs a solution with cost at most $(4+3\epsilon)mn(\frac{7}{6}-\delta) < mn\frac{26(7-6\delta)}{3(13-6\delta)}$.
	If an instance of case $(2)$ is given to $A$, it outputs a solution with cost greater than $mn(\frac{14}{3}-3\epsilon) > mn\frac{26(7-6\delta)}{3(13-6\delta)}$.
Hence, using $A$, we can distinguish between cases $(1)$ and $(2)$, which implies
that $\Poly = \NP$. 
\end{proof}

This establishes Theorem~\ref{thm:smfq_hardness}.\ref{hpart2}.

\section{Experiments}\label{sec:exp}
We present empirical results for the algorithms presented in our paper
on data sets constructed from real world data as well as on synthetically generated data sets.
The evaluations are performed on a 64-bit machine running 5.13.0-30-generic ubuntu kernel.
The machine has 32GB of RAM and Intel(R) Xeon(R) Silver 4210R 2.40GHz CPU with 20 cores.
Our algorithms are implemented and executed using Python 3.8.
We use IBM CPLEX solver academic version 
for solving the Integer Linear Program.

\subsection{Data sets}
Since we introduce the cost-based allocation under two-sided preference setting,
we do not have $\SMFQ$ instances. 
We generate $\SMFQ$ instances from the $\HR$ data sets using {\em cost functions} (see Section~\ref{sec:costfun}).
We use real $\HR$ data sets ({\bf R1, R2} and {\bf R3}) that contain 
the preferences of students (agents) and courses (programs) 
collected at an educational institute
over a period of three semesters. Each student ranks a subset of courses and
every course ranks a subset of students such that they are mutually acceptable.
The courses in these data sets have a fixed input quota associated with them. 

We generate synthetic $\HR$ data sets ({\bf S1, S2} and {\bf S3}) 
as follows.
The generator takes as input the number of students (agents), the number of programs (courses) and the length of preference list of
every student ($\ell_a$).
The quotas for the courses are assigned using a normalized uniform distribution $(0.0,1.0)$
such that the total quota is at least $\mid\hspace{-0.1cm}\AAA\hspace{-0.1cm}\mid$ and at most $2\mid\hspace{-0.1cm}\AAA\hspace{-0.1cm}\mid$.
The course popularity is also decided using a normalized uniform distribution $(0.0,1.0)$.
The courses for a fixed student are selected at random 
based on popularity and the ranks are assigned in the descending order of popularity.
The preference list of a course is generated using a random shuffling of the students who rank that course.

\subsection{Cost functions}\label{sec:costfun}
We convert the $\HR$ instances (both real and synthetic) using the cost functions that we have designed
and obtain the $\SMFQ$ instances. 
Recall that $len(p)$ denotes the length of preference list of program $p$
and $q(p)$ denotes the upper-quota of program $p$.
Our cost functions assign a cost to program $p$ that is directly proportional to $len(p)$
and inversely proportional to $q(p)$.
In practice, a program with high demand, that is, large value of $len(p)$ and
a small quota indicates that it is a {\em popular course} and hence is likely to have a high cost of matching an agent.
We have three cost functions as described below.
The {\bf \em median} cost function always generates an instance with two distinct costs.
Each of {\bf \em linear} and {\bf \em exponential} cost functions generate an instance with an arbitrary number
of distinct costs. Both the cost functions generate the same number of distinct costs. In the {\bf \em linear}
cost function, the absolute cost values are consecutive whereas in the {\bf \em exponential} cost function,
the absolute cost values span a larger range.
Let $ratio(p) = \frac{len(p)}{q(p)}$.
\begin{itemize}
	\item {\bf \em median(c). } 
		The cost function {\em median} is parameterized by a constant $c \geq 1$.
		Let $m$ be the median of the set of ratios $\{ratio(p) \mid p \in \BBB\}$. Set $c(p) = 0$ if $ratio(p) \leq m$, otherwise $c(p) = c$.
	\item {\bf \em linear. } Order programs using the distinct ratios $ratio(p)$ in increasing order and assign indices $0, 1, \ldots$ to the programs in that order. Then assign $c(p) = j$ if index of $p$ is $j$.
	\item {\bf \em exponential(c). } 
		The cost function {\em exponential} is parameterized by a constant $c \geq 2$.
		 Order programs using the distinct ratios $ratio(p)$ in increasing order and assign indices $0, 1, \ldots$ to the programs in that order. Then assign $c(p) = c^j$ if index of $p$ is $j$.
\end{itemize}

In our experiments, we have selected $c=10$ for the {\bf \em median} cost function and $c=2$ for the {\bf \em exponential} cost function.
{Additionally, we have evaluated our algorithms for the instances generated by setting   $c = 5$ and $c = 100$ for the {\bf \em median} cost function and $c = 4$ and $c=5$ for the {\bf \em exponential} cost function. 
We omit the detailed results for these settings since the results follow a similar pattern as the results we have presented for our selected choice of $c$.}

\subsection{Properties of the generated $\SMFQ$ instances}
Table~\ref{tab:datasets} shows the properties of the generated $\SMFQ$ instances.
For an $\HR$ instance, $M_s$ denotes a stable matching in that instance.
For each input $\HR$ instance, we instantiate a $\SMFQ$ instance using three cost functions described in Section~\ref{sec:costfun}
to get a total of eighteen $\SMFQ$ instances.
For a fixed $\HR$ instance, the number of agents ($\mid\hspace{-0.1cm}\AAA\hspace{-0.1cm}\mid$), the number of programs ($\mid\hspace{-0.1cm}\BBB\hspace{-0.1cm}\mid$),
the length of the longest preference list of an agent and a program (respectively, $\ell_a$ and $\ell_p$)
and the size of a stable matching are indicated.

For each $\SMFQ$ instantiation, the column {\bf tuple} denotes three values --
the number of distinct costs, the minimum cost value and the maximum cost value in that order.
	The optimal solution for $\SMFQUSET$ problem on the particular instance is obtained using the 
$\SMFQUSET$ Integer Linear Program presented 
in Section~\ref{sec:c1c2} that is solved using IBM CPLEX solver.
For each instance, we report the time (in seconds) required for the CPLEX solver 
and the cost of the optimal matching. Recall that the optimal matchings are $\AAA$-perfect.

\begin{table}[!ht]
{
	\fontsize{0.0001}{0.0001}\selectfont
	\def\arraystretch{30}
	\setlength\tabcolsep{1.4pt}
	\begin{tabular}{ccccccP{0.08\textwidth}P{0.07\textwidth}P{0.05\textwidth}P{0.08\textwidth}P{0.06\textwidth}P{0.06\textwidth}P{0.15\textwidth}P{0.06\textwidth}P{0.08\textwidth}}\toprule
		\multirow{3}{*}{} & \multirow{3}{*}[-8pt]{$\mid\hspace{-0.08cm}\AAA\hspace{-0.08cm}\mid$} & \multirow{3}{*}[-8pt]{$\mid\hspace{-0.08cm}\BBB\hspace{-0.08cm}\mid$} & \multirow{3}{*}[-8pt]{$\ell_a$} & \multirow{3}{*}[-8pt]{$\ell_p$} &\multirow{3}{*}[-8pt]{$\mid\hspace{-0.08cm}M_s\hspace{-0.08cm}\mid$} & \multicolumn{3}{c}{\bf median} & \multicolumn{3}{c}{\bf linear} & \multicolumn{3}{c}{\bf exponential} \\\cmidrule(lr){7-15}
		& & & & & & \multirow{2}{*}{\bf tuple} & \multicolumn{2}{c}{\bf LP $\OPT$} & \multirow{2}{*}{\bf tuple} & \multicolumn{2}{c}{\bf LP $\OPT$} & \multirow{2}{*}{\bf tuple} & \multicolumn{2}{c}{\bf LP $\OPT$}\\\cmidrule(lr){8-9}\cmidrule(lr){11-12}\cmidrule(lr){14-15}
		& & & & & & & {\bf time (sec)} & {\bf cost} & & {\bf time (sec)} & {\bf cost} & & {\bf time (sec)} & {\bf cost} \\\midrule
		{\bf R1} & 655 & 14 & 14 & 386 & 487 & (2,0,10) & 2.92 & 5640 & (13,0,12) & 2.79 & 6331 & (13,1,2.44$\times10^8$) & 2.81 & 1.26$\times10^6$ \\\midrule
		{\bf R2} & 729 & 16 & 16 & 513 & 675 & (2,0,10) & 22.23 & 6230 & (16,0,15) & 14.10 & 8442 & (16,1,3.05$\times10^{10}$) & 13.71 & 8.19$\times10^6$ \\\midrule
		{\bf R3} & 483 & 18 & 18 & 412 & 482 & (2,0,10) & 264.84 & 2380 & (18,0,17) & 2550 & 4452 & (18,1,7.62$\times10^{11}$) & 53.42 & 4.24$\times10^6$ \\\midrule
		{\bf S1} & 500 & 20 & 5 & 203 & 430 & (2,0,10) & 5.64 & 220 & (20,0,19) & 19.86 & 2968 & (20,1,5.24$\times10^5$) & 6.13 & 2.44$\times10^5$ \\\midrule
		{\bf S2} & 750 & 35 & 5 & 207 & 604 & (2,0,10) & 8155 & 4550 & (35,0,34) & 3827 & 14472 & (35,1,1.71$\times10^{10}$) & 9.85 & 5.28$\times10^9$ \\\midrule
		{\bf S3} & 1000 & 50 & 5 & 193 & 840 & (2,0,10) & 2257 & 3740 & (49,0,48) & 11372 & 24797 & (49,1,2.81$\times10^{14}$) & 18.16 & 3.82$\times10^{12}$ \\\bottomrule
\end{tabular}
}
\caption{Properties of the $\SMFQ$ instances}
\label{tab:datasets}
\end{table}

\subsection{Evaluation Parameters}
We implement the $\ell_p$-approximation algorithms (Algorithm~\ref{algo:algo_mincostset_arb} and $\sf ALG$)
and the $\mid\hspace{-0.1cm}\BBB\hspace{-0.1cm}\mid$-approximation algorithm in Section~\ref{sec:papprox}.
We also implement the stable matching algorithms~\cite{GS62} to
compute the unique agent-optimal (denoted as {\em A-opt}) and program-optimal (denoted as
{\em P-opt}) stable matching.
In the implementation of the $\mid\hspace{-0.1cm}\BBB\hspace{-0.1cm}\mid$-approximation algorithm we compute the {\em A-opt} stable matching.
Table~\ref{tab:expres} and Table~\ref{tab:expressynth} report our results on the $\SMFQ$ instances
generated from the real and synthetic data sets respectively.
The row LP $\OPT$ denotes the $\SMFQUSET\ \OPT$.
The next three rows 
respectively denote the $\ell_p$-approximation algorithms $\sf ALG$ and
Algorithm~\ref{algo:algo_mincostset_arb} and the $\mid\hspace{-0.1cm}\BBB\hspace{-0.1cm}\mid$-approximation algorithm in Section~\ref{sec:papprox}.
For each $\SMFQ$ instance, and for every approximation algorithm, we evaluate various parameters which can be classified
into the following three sets.

\noindent{\bf Parameters specific to the  algorithms. }
Recall that the absolute time and cost values for the optimal solution computed using CPLEX 
solver are reported in Table~\ref{tab:datasets}. We report the following parameters
specific to the algorithms.
\begin{enumerate}
	\item {\bf \em time ratio}: the relative time required for execution. The relative time is
		reported as the ratio {(truncated to three decimal places)} of the absolute time value for the algorithm and the absolute time value for 
		obtaining the optimal solution.
	\item {\bf \em approx. ratio}: the approximation guarantee. The approximation guarantee is the ratio
		{(truncated to three decimal places)}
		of the absolute cost of the output matching obtained using the algorithm and the 
		absolute cost of the optimal solution.
\end{enumerate}
{For instance, $\sf ALG$ executed on the $\SMFQ$ instance generated from the real data set {\bf R1} using {\bf \em median} cost function
takes absolute time equal to 0.295 times the time taken by the CPLEX solver (that is, 2.92 seconds as seen from Table~\ref{tab:datasets}) and the absolute cost of the solution is 1.005 times the optimal cost (that is, 5640 as seen from Table~\ref{tab:datasets}).}

\noindent{\bf Parameters indicating the quality of the output matching. }
We measure the following parameters that capture the quality of the output matching in terms of agent satisfaction.
\begin{enumerate}
	\item {\bf \em avg-rank}: the average rank at which an agent is matched.
	\item {\bf \em rank-1\%} or {\bf \em top-3\%}: the percentage of agents matched to their rank-1 program for 
		instances generated from the synthetic data sets or within their top-3 ranks for instances
		generated from the real data sets.
		We choose to measure {\bf \em rank-1\%} for 
the synthetic data sets
		because by design $\ell_a = 5$ for these data sets.
\setcounter{savedctr}{\value{enumi}}
\end{enumerate}
It is well-known~\cite{GI89} that for any $\HR$ instance every stable matching matches the same
set of agents. Furthermore, for any agent $a$,
the matched program $M_s(a)$ for any stable matching $M_s$ is never 
worse than the program matched to it in the unique {\em P-opt} stable matching
and is never better than the program matched to it in the unique {\em A-opt} stable matching.
Since the $\SMFQ$ instances do not have quotas, an agent may be matched to a program
that falls outside this {\em range}.
We measure the following two parameters that capture this. To make the comparison fair, the agents considered for these parameters are only the ones that
are matched in a stable matching. This is because the output for the $\SMFQUSET$ is always $\AAA$-perfect. Thus, an agent unmatched in a stable matching 
is trivially matched to a better-preferred program in any output matching.
\begin{enumerate}
\setcounter{enumi}{\value{savedctr}}
\item {\bf \em $\lpref$ {\em P-opt} stable\%}:	the percentage of agents matched in a stable matching that prefer their {\em P-opt}
	partner over the program matched to them in the output matching.
\item {\bf \em $\mpref$ {\em A-opt} stable\%}:
		the percentage of agents matched in a stable
		matching that prefer the program matched to them in the output matching over 
		their {\em A-opt} partner.
\end{enumerate}

If $M_{\BBB}$ and $M_{\AAA}$ are the {\em P-opt} and {\em A-opt} stable matchings respectively and $M$ is the output matching then
\begin{eqnarray*}
	\lpref \text{{\em P-opt} stable\%} = \frac{100\cdot\mid\hspace{-0.1cm}\{a \mid M_{\BBB}(a) \neq \bot, M_{\BBB}(a) \mpref_a M(a)\}\hspace{-0.1cm}\mid}{\mid\hspace{-0.1cm}M_{\BBB}\hspace{-0.1cm}\mid}\\\\
	\mpref \text{{\em A-opt} stable\%} = \frac{100\cdot\mid\hspace{-0.1cm}\{a \mid M_{\AAA}(a) \neq \bot, M(a) \mpref_a M_{\AAA}(a)\}\hspace{-0.1cm}\mid}{\mid\hspace{-0.1cm}M_{\AAA}\hspace{-0.1cm}\mid}
\end{eqnarray*}

\begin{table}[!ht]
	{
	\fontsize{0.0001}{0.0001}\selectfont
	\def\arraystretch{0}
	\setlength\tabcolsep{1pt}
	\begin{tabular}{P{0.12\textwidth}P{0.12\textwidth}P{0.055\textwidth}P{0.065\textwidth}P{0.07\textwidth}P{0.07\textwidth}P{0.075\textwidth}P{0.075\textwidth}P{0.075\textwidth}P{0.075\textwidth}P{0.075\textwidth}}\toprule
		& \multirow{2}{*}{\bf Algo} & \multicolumn{2}{P{0.12\textwidth}}{\bf Algorithm parameters} & \multicolumn{4}{P{0.29\textwidth}}{\bf Quality of output matching} & \multicolumn{3}{P{0.225\textwidth}}{\bf Comparison to $\HR$ model} \\\cmidrule(lr){3-4}\cmidrule(lr){5-8}\cmidrule(lr){9-11}
		& & {\bf time ratio ($\downarrow$)} & {\bf approx. ratio ($\downarrow$)} & {\bf avg-rank ($\downarrow$)} & {\bf top-3\% ($\uparrow$)} & {\bf $\lpref$ {\em P-opt} stable\% ($\downarrow$)} & {\bf $\mpref$ {\em A-opt} stable\% ($\uparrow$)} & {\bf bp\% ($\downarrow$)} & {\bf ba\% ($\downarrow$)} & {\bf vio\% ($\downarrow$)} \\\midrule

		\multicolumn{11}{c}{} \\
		\multicolumn{11}{c}{\multirow{2}{*}[3pt]{\bf Data set R1}}\\\midrule
		\multirow{4}{*}[-12pt]{\bf median} & LP $\OPT$ & 1 & 1 & 1.056 & 99.847 & 0.205 & 60.78 & 1.672 & 3.969 & 230.714\\\cmidrule(lr){2-11}
		& $\sf ALG$ & 0.295 & 1.005 & 1 & 100 & 0 & 61.602 & 0 & 0 & 221.429\\\cmidrule(lr){2-11}
		& Algo.~\ref{algo:algo_mincostset_arb} & 0.414 & 1.002 & 1.005 & 100 & 0 & 61.396 & 0.049 & 0.153 & 220\\\cmidrule(lr){2-11}
		& $\mid\hspace{-0.08cm}\BBB\hspace{-0.08cm}\mid$-approx & 0.385 & 1.004 & 1.002 & 100 & 0 & 61.602 & 0 & 0 & 220.714\\\midrule

		\multirow{4}{*}[-12pt]{\bf linear} &  LP $\OPT$ & 1 & 1 & 1.055 & 100 & 0 & 60.78 & 1.721 & 3.969 & 232.143\\\cmidrule(lr){2-11}
		& $\sf ALG$ & 0.3 & 1.008 & 1 & 100 & 0 & 61.602 & 0 & 0 & 221.429\\\cmidrule(lr){2-11}
		& Algo.~\ref{algo:algo_mincostset_arb} & 0.432 & 1.002 & 1.043 & 100 & 0 & 60.986 & 1.327 & 4.122 & 226.429\\\cmidrule(lr){2-11}
		& $\mid\hspace{-0.08cm}\BBB\hspace{-0.08cm}\mid$-approx & 0.392 & 1.007 & 1.002 & 100 & 0 & 61.602 & 0 & 0 & 220.714\\\midrule

		\multirow{4}{*}[-12pt]{\bf exponential} & LP $\OPT$ & 1 & 1 & 1.044 & 100 & 0 & 60.78 & 1.377 & 2.901 & 230\\\cmidrule(lr){2-11}
		& $\sf ALG$ & 0.3 & 1.012 & 1 & 100 & 0 & 61.602 & 0 & 0 & 221.429\\\cmidrule(lr){2-11}
		& Algo.~\ref{algo:algo_mincostset_arb} & 0.428 & 1.003 & 1.043 & 100 & 0 & 60.986 & 1.327 & 4.122 & 226.429\\\cmidrule(lr){2-11}
		& $\mid\hspace{-0.08cm}\BBB\hspace{-0.08cm}\mid$-approx & 0.364 & 1.009 & 1.002 & 100 & 0 & 61.602 & 0 & 0 & 220.714\\\midrule

		\multicolumn{11}{c}{} \\
		\multicolumn{11}{c}{\multirow{2}{*}[3pt]{\bf Data set R2}}\\\midrule
		\multirow{4}{*}[-12pt]{\bf median} & LP $\OPT$ & 1 & 1 & 1.418 & 97.119 & 2.519 & 49.778 & 3.443 & 14.815 & 160\\\cmidrule(lr){2-11}
		& $\sf ALG$ & 0.049 & 1.048 & 1 & 100 & 0 & 61.926 & 0 & 0 & 168.333\\\cmidrule(lr){2-11}
		& Algo.~\ref{algo:algo_mincostset_arb} & 0.077 & 1.042 & 1.049 & 100 & 0.593 & 61.185 & 0.499 & 2.606 & 172.778\\\cmidrule(lr){2-11}
		& $\mid\hspace{-0.08cm}\BBB\hspace{-0.08cm}\mid$-approx & 0.071 & 1.043 & 1.134 & 100 & 0 & 58.667 & 0 & 0 & 120.833\\\midrule

		\multirow{4}{*}[-12pt]{\bf linear} & LP $\OPT$ & 1 & 1 & 1.431 & 93.553 & 0.444 & 52 & 5.02 & 18.519 & 163.889\\\cmidrule(lr){2-11}
		& $\sf ALG$ & 0.077 & 1.052 & 1 & 100 & 0 & 61.926 & 0 & 0 & 168.333\\\cmidrule(lr){2-11}
		& Algo.~\ref{algo:algo_mincostset_arb} & 0.115 & 1.006 & 1.443 & 93.416 & 0.444 & 51.852 & 5.072 & 18.93 & 165\\\cmidrule(lr){2-11}
		& $\mid\hspace{-0.08cm}\BBB\hspace{-0.08cm}\mid$-approx & 0.122 & 1.013 & 1.134 & 100 & 0 & 58.667 & 0 & 0 & 120.833\\\midrule

		\multirow{4}{*}[-12pt]{\bf exponential} & LP $\OPT$ & 1 & 1 & 1.163 & 100 & 0.444 & 57.926 & 0.447 & 2.195 & 125.417\\\cmidrule(lr){2-11}
		& $\sf ALG$ & 0.08 & 1.348 & 1 & 100 & 0 & 61.926 & 0 & 0 & 168.333\\\cmidrule(lr){2-11}
		& Algo.~\ref{algo:algo_mincostset_arb} & 0.12 & 1.051 & 1.443 & 93.416 & 0.444 & 51.852 & 5.072 & 18.93 & 165\\\cmidrule(lr){2-11}
		& $\mid\hspace{-0.08cm}\BBB\hspace{-0.08cm}\mid$-approx & 0.125 & 1.006 & 1.134 & 100 & 0 & 58.667 & 0 & 0 & 120.833\\\midrule

		\multicolumn{11}{c}{} \\
		\multicolumn{11}{c}{\multirow{2}{*}[3pt]{\bf Data set R3}}\\\midrule
		\multirow{4}{*}[-12pt]{\bf median} & LP $\OPT$ & 1 & 1 & 2.427 & 78.054 & 34.232 & 20.332 & 12.671 & 37.474 & 130.476\\\cmidrule(lr){2-11}
		& $\sf ALG$ & 0.002 & 1.689 & 1.112 & 100 & 10.788 & 29.046 & 1.118 & 10.766 & 96.364\\\cmidrule(lr){2-11}
		& Algo.~\ref{algo:algo_mincostset_arb} & 0.003 & 1 & 2.41 & 78.054 & 34.232 & 20.332 & 12.505 & 37.474 & 127.619\\\cmidrule(lr){2-11}
		& $\mid\hspace{-0.08cm}\BBB\hspace{-0.08cm}\mid$-approx & 0.004 & 1.723 & 1.389 & 96.273 & 0 & 11.203 & 0 & 0 & 16.364\\\midrule

		\multirow{4}{*}[-12pt]{\bf linear} & LP $\OPT$ & 1 & 1 & 3.211 & 54.244 & 54.149 & 1.867 & 19.648 & 65.839 & 137.273\\\cmidrule(lr){2-11}
		& $\sf ALG$ & 0 & 1.117 & 2.503 & 61.491 & 48.34 & 6.224 & 15.031 & 58.592 & 137.576\\\cmidrule(lr){2-11}
		& Algo.~\ref{algo:algo_mincostset_arb} & 0 & 1.033 & 2.754 & 59.213 & 49.378 & 5.394 & 16.625 & 60.041 & 162.727\\\cmidrule(lr){2-11}
		& $\mid\hspace{-0.08cm}\BBB\hspace{-0.08cm}\mid$-approx & 0 & 1.308 & 1.28 & 98.965 & 0 & 16.183 & 0 & 0 & 29.091\\\midrule

		\multirow{4}{*}[-12pt]{\bf exponential} & LP $\OPT$ & 1 & 1 & 2.87 & 59.006 & 49.793 & 5.394 & 17.557 & 60.041 & 155.455\\\cmidrule(lr){2-11}
& $\sf ALG$ & 0.009 & 1.506 & 2.503 & 61.491 & 48.34 & 6.224 & 15.031 & 58.592 & 137.576\\\cmidrule(lr){2-11}
& Algo.~\ref{algo:algo_mincostset_arb} & 0.013 & 1.001 & 2.754 & 59.213 & 49.378 & 5.394 & 16.625 & 60.041 & 162.727\\\cmidrule(lr){2-11}
& $\mid\hspace{-0.08cm}\BBB\hspace{-0.08cm}\mid$-approx & 0.02 & 2.105 & 1.174 & 100 & 0.83 & 20.954 & 0.228 & 2.277 & 52.727\\\bottomrule
		
	\end{tabular}
	}
	\caption{Evaluation of the instances generated from real data sets.
	}
	\label{tab:expres}
\end{table}

\begin{table}[!ht]
	{
	\fontsize{0.0001}{0.0001}\selectfont
	\def\arraystretch{0}
	\setlength\tabcolsep{1pt}
	\begin{tabular}{P{0.12\textwidth}P{0.12\textwidth}P{0.055\textwidth}P{0.065\textwidth}P{0.07\textwidth}P{0.07\textwidth}P{0.075\textwidth}P{0.075\textwidth}P{0.075\textwidth}P{0.075\textwidth}P{0.075\textwidth}}\toprule
		& \multirow{2}{*}{\bf Algo} & \multicolumn{2}{P{0.12\textwidth}}{\bf Algorithm parameters} & \multicolumn{4}{P{0.29\textwidth}}{\bf Quality of output matching} & \multicolumn{3}{P{0.225\textwidth}}{\bf Comparison to $\HR$ model} \\\cmidrule(lr){3-4}\cmidrule(lr){5-8}\cmidrule(lr){9-11}

		& & {\bf time ratio ($\downarrow$)} & {\bf approx. ratio ($\downarrow$)} & {\bf avg-rank ($\downarrow$)} & {\bf rank-1\% ($\uparrow$)} & {\bf $\lpref$ {\em P-opt} stable\% ($\downarrow$)} & {\bf $\mpref$ {\em A-opt} stable\% ($\uparrow$)} & {\bf bp\% ($\downarrow$)} & {\bf ba\% ($\downarrow$)} & {\bf vio\% ($\downarrow$)} \\\midrule

		\multicolumn{11}{c}{} \\
		\multicolumn{11}{c}{\multirow{2}{*}[3pt]{\bf Data set S1}}\\\midrule
		\multirow{4}{*}[-12pt]{\bf median} & LP $\OPT$ & 1 & 1 & 2.594 & 14.8 & 18.372 & 48.14 & 39.2 & 85 & 136.932\\\cmidrule(lr){2-11}
		& $\sf ALG$ & 0.091 & 19 & 1.018 & 98.4 & 0.233 & 86.977 & 0.45 & 1.6 & 311.304\\\cmidrule(lr){2-11}
		& Algo.~\ref{algo:algo_mincostset_arb} & 0.114 & 16.545 & 1.204 & 84.2 & 0.698 & 81.395 & 1.1 & 3.8 & 275.652\\\cmidrule(lr){2-11}
		& $\mid\hspace{-0.08cm}\BBB\hspace{-0.08cm}\mid$-approx & 0.154 & 2.318 & 2.456 & 18 & 12.093 & 49.07 & 22.85 & 67.8 & 173.485\\\midrule

		\multirow{4}{*}[-12pt]{\bf linear} & LP $\OPT$ & 1 & 1 & 2.724 & 14.4 & 22.326 & 41.86 & 41.45 & 85.4 & 172.662\\\cmidrule(lr){2-11}
		& $\sf ALG$ & 0.025 & 1.641 & 1.62 & 50.4 & 3.721 & 72.326 & 15.5 & 49.6 & 204.225\\\cmidrule(lr){2-11}
		& Algo.~\ref{algo:algo_mincostset_arb} & 0.031 & 1.348 & 2.048 & 31.4 & 5.349 & 58.14 & 18.65 & 54.4 & 173.944\\\cmidrule(lr){2-11}
		& $\mid\hspace{-0.08cm}\BBB\hspace{-0.08cm}\mid$-approx & 0.05 & 1.464 & 2.616 & 18.8 & 0 & 36.744 & 0 & 0 & 54.58\\\midrule

		\multirow{4}{*}[-12pt]{\bf exponential} & LP $\OPT$ & 1 & 1 & 2.576 & 14.6 & 17.907 & 48.14 & 39.2 & 85.2 & 136.932\\\cmidrule(lr){2-11}
		& $\sf ALG$ & 0.083 & 7.147 & 1.62 & 50.4 & 3.721 & 72.326 & 15.5 & 49.6 & 204.225\\\cmidrule(lr){2-11}
		& Algo.~\ref{algo:algo_mincostset_arb} & 0.102 & 3.097 & 2.048 & 31.4 & 5.349 & 58.14 & 18.65 & 54.4 & 173.944\\\cmidrule(lr){2-11}
		& $\mid\hspace{-0.08cm}\BBB\hspace{-0.08cm}\mid$-approx & 0.138 & 2.282 & 2.334 & 21 & 6.744 & 52.093 & 28.15 & 72.4 & 137.58\\\midrule

		\multicolumn{11}{c}{} \\
		\multicolumn{11}{c}{\multirow{2}{*}[3pt]{\bf Data set S2}}\\\midrule
		\multirow{4}{*}[-12pt]{\bf median} & LP $\OPT$ & 1 & 1 & 2.747 & 21.467 & 17.219 & 49.007 & 30.233 & 70.667 & 185.022\\\cmidrule(lr){2-11}
		& $\sf ALG$ & 0 & 1.514 & 1.092 & 96.133 & 1.821 & 84.272 & 2.3 & 3.867 & 375.177\\\cmidrule(lr){2-11}
		& Algo.~\ref{algo:algo_mincostset_arb} & 0 & 1.459 & 1.199 & 92 & 2.98 & 80.795 & 4.8 & 7.467 & 439.831\\\cmidrule(lr){2-11}
		& $\mid\hspace{-0.08cm}\BBB\hspace{-0.08cm}\mid$-approx & 0 & 1.411 & 1.799 & 45.2 & 0 & 70.861 & 0 & 0 & 161.338\\\midrule

		\multirow{4}{*}[-12pt]{\bf linear} & LP $\OPT$ & 1 & 1 & 2.285 & 35.733 & 9.934 & 55.629 & 29.433 & 63.333 & 194.286\\\cmidrule(lr){2-11}
		& $\sf ALG$ & 0 & 1.086 & 1.647 & 52.533 & 4.305 & 71.358 & 16.167 & 47.467 & 243.814\\\cmidrule(lr){2-11}
		& Algo.~\ref{algo:algo_mincostset_arb} & 0 & 1.051 & 1.803 & 48.267 & 5.298 & 66.887 & 18.367 & 49.333 & 197.788\\\cmidrule(lr){2-11}
		& $\mid\hspace{-0.08cm}\BBB\hspace{-0.08cm}\mid$-approx & 0.001 & 1.248 & 1.809 & 44.267 & 0 & 71.192 & 0 & 0 & 162.082\\\midrule

		\multirow{4}{*}[-12pt]{\bf exponential} & LP $\OPT$ & 1 & 1 & 1.816 & 48.8 & 5.795 & 66.06 & 18.867 & 49.2 & 218.627\\\cmidrule(lr){2-11}
		& $\sf ALG$ & 0.117 & 1.411 & 1.647 & 52.533 & 4.305 & 71.358 & 16.167 & 47.467 & 243.814\\\cmidrule(lr){2-11}
		& Algo.~\ref{algo:algo_mincostset_arb} & 0.142 & 1.006 & 1.803 & 48.267 & 5.298 & 66.887 & 18.367 & 49.333 & 197.788\\\cmidrule(lr){2-11}
		& $\mid\hspace{-0.08cm}\BBB\hspace{-0.08cm}\mid$-approx & 0.181 & 2.309 & 1.792 & 47.2 & 3.974 & 68.709 & 16.867 & 49.6 & 184.711\\\midrule

		\multicolumn{11}{c}{} \\
		\multicolumn{11}{c}{\multirow{2}{*}[3pt]{\bf Data set S3}}\\\midrule
		\multirow{4}{*}[-12pt]{\bf median} & LP $\OPT$ & 1 & 1 & 2.853 & 11.9 & 18.929 & 37.857 & 40.3 & 86.5 & 165.677\\\cmidrule(lr){2-11}
		& $\sf ALG$ & 0.001 & 2.447 & 1.011 & 99.4 & 0.238 & 79.762 & 0.275 & 0.6 & 388.136\\\cmidrule(lr){2-11}
		& Algo.~\ref{algo:algo_mincostset_arb} & 0.001 & 2.313 & 1.124 & 92.1 & 0.714 & 75.714 & 2.15 & 5.6 & 319.712\\\cmidrule(lr){2-11}
		& $\mid\hspace{-0.08cm}\BBB\hspace{-0.08cm}\mid$-approx & 0.001 & 1.62 & 2.216 & 29.3 & 1.31 & 50 & 9.4 & 33.2 & 133.143\\\midrule

		\multirow{4}{*}[-12pt]{\bf linear} & LP $\OPT$ & 1 & 1 & 2.409 & 31.4 & 10.595 & 43.571 & 26.775 & 61.8 & 110.462\\\cmidrule(lr){2-11}
		& $\sf ALG$ & 0 & 1.137 & 1.363 & 75.2 & 2.619 & 69.762 & 9.075 & 24.8 & 229.688\\\cmidrule(lr){2-11}
		& Algo.~\ref{algo:algo_mincostset_arb} & 0 & 1.111 & 1.529 & 65.7 & 2.857 & 63.929 & 10.475 & 27.6 & 198.175\\\cmidrule(lr){2-11}
		& $\mid\hspace{-0.08cm}\BBB\hspace{-0.08cm}\mid$-approx & 0 & 1.243 & 2 & 38 & 0 & 56.429 & 0 & 0 & 105.519\\\midrule

		\multirow{4}{*}[-12pt]{\bf exponential} & LP $\OPT$ & 1 & 1 & 1.664 & 60.4 & 5.238 & 61.786 & 14.95 & 37.3 & 234.667\\\cmidrule(lr){2-11}
		& $\sf ALG$ & 0.111 & 1.591 & 1.363 & 75.2 & 2.619 & 69.762 & 9.075 & 24.8 & 229.688\\\cmidrule(lr){2-11}
		& Algo.~\ref{algo:algo_mincostset_arb} & 0.124 & 1.284 & 1.529 & 65.7 & 2.857 & 63.929 & 10.475 & 27.6 & 198.175\\\cmidrule(lr){2-11}
		& $\mid\hspace{-0.08cm}\BBB\hspace{-0.08cm}\mid$-approx & 0.134 & 1.832 & 1.43 & 70.4 & 2.024 & 67.738 & 9.4 & 25.8 & 228.049\\\bottomrule

	\end{tabular}
	}
	\caption{Evaluation of the instances generated from synthetic data sets. 
	}
	\label{tab:expressynth}
\end{table}

\noindent{\bf Parameters related to the corresponding $\HR$ instance. }
Since we derive our $\SMFQ$ instances from $\HR$ instances, we evaluate the output matching with respect to
the original $\HR$ instance.
Recall that our output matchings are envy-free.
However, an agent may prefer program $p$ over its matched partner in the output matching
such that $p$ is under-subscribed w.r.t. the input quota in the corresponding $\HR$ instance.
Note that such blocking pairs are not relevant in the $\SMFQ$ instance since input quotas are absent.
The following two parameters capture this.

\begin{enumerate}
	\item {\bf \em bp\%}: the percentage of the number of blocking pairs. Let $bp$ denote the number of blocking pairs
		w.r.t. the output matching $M$ and the input quotas and $m$ denote the number of edges in the instance. Then
		$bp\% = \frac{bp\cdot 100}{m-\mid M\mid} = \frac{bp \cdot 100}{m-\mid\AAA\mid}$ since $M$ is $\AAA$-perfect.
	\item {\bf \em ba\%}: the number of blocking agents. Let $ba$ denote the number of blocking agents
		w.r.t. the output matching $M$ and the input quotas. Then
		$ba\% = \frac{ba\cdot 100}{\mid\AAA\mid}$.
\setcounter{savedctr}{\value{enumi}}
\end{enumerate}
In order to achieve $\AAA$-perfectness, the output matching in the $\SMFQ$ instance
derived from the $\HR$ instance is expected to violate the input quotas, 
that is, $\mid\hspace{-0.1cm}M(p)\hspace{-0.1cm}\mid > q(p)$ for
some program $p$.
This is the price that the $\SMFQ$ setting pays in order to achieve
$\AAA$-perfectness.
We compute the percentage of total {\em violation} of the output matching.
\begin{enumerate}
\setcounter{enumi}{\value{savedctr}}
\item {\bf \em vio\%}: Let {\em violation} of program $p$ be
	the difference between $\mid\hspace{-0.1cm} M(p)\hspace{-0.1cm}\mid$ and $q(p)$ if positive, undefined otherwise.
The total violation is the sum of individual violation of the programs whose violation is defined.
The total input quota is the sum of input quotas of programs whose violation is defined.
Then $vio\% = \frac{\text{total violation}\cdot 100}{\text{total input quota}}$.
\end{enumerate}

\subsection{Summary of the empirical evaluations}
As mentioned earlier, Table~\ref{tab:expres} and Table~\ref{tab:expressynth} report our results.
In each column, the symbols $(\uparrow)$ and $(\downarrow)$ respectively indicate whether the larger or smaller value
of the corresponding parameter is better.

We make the following observations from our results.

\begin{itemize}
	\item The approximation ratio for all three algorithms 
		($\sf ALG$, Algorithm~\ref{algo:algo_mincostset_arb} and the 
		$\mid\hspace{-0.1cm}\BBB\hspace{-0.1cm}\mid$-approximation algorithm) is upper bounded by 2.5
		except for the data set {\bf S1}.
		Although, for {\bf S1} and the {\bf \em median} cost function, the absolute value of the ratio is large (around 16 for 
Algorithm~\ref{algo:algo_mincostset_arb} and around 19 for $\sf ALG$) we note
that the performance is significantly better than the theoretical guarantee. Furthermore even on this instance, the
$\mid\hspace{-0.14cm}\BBB\hspace{-0.14cm}\mid$-approx ratio is bounded by 2.5.
In all our data sets, among the two $\ell_p$-approximation algorithms, 
Algorithm~\ref{algo:algo_mincostset_arb} always performs better than $\sf ALG$ in terms of the approximation ratio. 
However, as seen in Section~\ref{sec:compare} neither of them is strictly better than the other.

\item		All our algorithms are simple to implement and the maximum speedup is of the order of $10^3$. The speedup is inverse of
column {\bf time ratio} in Table~\ref{tab:expres} and Table~\ref{tab:expressynth}.
It is seen that there is a variation in the speedups obtained on different instances -- this is attributed to the 
variation in the absolute time the LP \OPT\ takes on specific data sets.
For example in the data set
		 {\bf S3} and {\bf \em linear} cost function, the CPLEX
		solver takes a prohibitively large time (more than 3 hours).
	\item The average rank and the percentage of agents matched to their top choices are parameters of practical interest
in real-world applications like course allocation, or job assignment. Since the output matching in the $\SMFQ$ setting is $\AAA$-perfect, it 
is expected that average rank is small (smaller value is better) and percentage of agents matched to top choices is large. 
We observe that for real data sets all the approximation algorithms perform very well with respect to these two parameters.
For synthetic data sets, there is variation across different algorithms. We note that for a fixed instance
		the LP $\OPT$ and $\mid\hspace{-0.1cm}\BBB\hspace{-0.1cm}\mid$-approx have similar values and $\sf ALG$ and Algorithm~\ref{algo:algo_mincostset_arb} outperform them on most instances with respect to the two parameters.
		\item The parameters  {\bf \em $\mpref$ A-opt stable\%}  and  {\bf \em
		$\lpref$ P-opt stable\%}  allow us to compare our output with any stable matching in the original $\HR$ instance. 
For all data sets except {\bf R3} we observe that we consistently have large values for the first parameter and small values for the second.
For the data sets {\bf R3}, the underlying $\HR$ instance admits a close to $\AAA$-perfect stable matching. This possibly explains the 
unusual values for the two parameters for all cost functions corresponding to {\bf R3}.
	\item The percentage of blocking pairs and blocking agents are parameters of interest when translating the
		solution from $\SMFQ$ setting back to the underlying $\HR$ instance. 
		We note that for real instances (except {\bf R3}), the percentage of 
blocking pairs is bounded by 5\% whereas percentage blocking agents is bounded by 20\%. We see a variation of these parameters
across different algorithms (especially in synthetic data sets -- see for instance ${\bf S1}$ where $\sf ALG$, Algorithm~\ref{algo:algo_mincostset_arb} perform significantly better as compared to LP $\OPT$ and $\mid\hspace{-0.1cm}\BBB\hspace{-0.1cm}\mid$-approx). 
		As mentioned earlier, the percentage of violation is the price the $\SMFQ$ 
		setting pays to achieve $\AAA$-perfectness. It is noted that LP $\OPT$ also suffers large violation of the quotas
to achieve $\AAA$-perfectness. We observe that in almost all instances $\mid\hspace{-0.1cm}\BBB\hspace{-0.1cm}\mid$-approximation algorithm
achieves the minimum value of percentage total violation.
\end{itemize}

From our empirical evaluations, we conclude that in terms of approximation guarantee as well as other parameters of interest
$\mid\hspace{-0.1cm}\BBB\hspace{-0.1cm}\mid$-approximation algorithm seems a practically appealing algorithm. 
An added advantage of the algorithm is that the same algorithm provides an exact solution for the
$\minmax$ problem on the same $\SMFQ$ instance.

\section{Concluding Remarks}\label{sec:disc}
In this work we propose cost-based allocation for the bipartite matching problem under the two-sided preference lists setting.
We propose and investigate the problem of computing $\AAA$-perfect envy-free matchings under two optimization criteria, namely $\minmax$ and $\SMFQUSET$ problems and prove a sharp contrast in their complexity.
We present an efficient algorithm for $\minmax$ problem and 
three approximation algorithms for $\SMFQUSET$ on general instances.
Our empirical evaluations indicate that the approximation algorithms perform reasonably well even on the parameters like percentage of blocking agents
and average rank, for which the optimization is not inherent in the problem. 
We present a novel LP for $\SMFQUSET$ and a primal dual approximation algorithm for a special hard case.
We remark that our $\NP$-hardness result implies that the $\SMFQUSET$ problem is para-$\NP$-hard when
parameterized on the length of an agent's preference list or the number of distinct costs in the instance.

We note that for the $\SMFQUSET$ problem, there is a gap between the upper bound and the lower bound
shown in this work. A specific open direction is to show that $\SMFQUSET$ does not admit a constant factor approximation.
It will be also interesting to see if the $\NP$-hardness holds even when the instance has distinct costs.
Another open direction is to design cost functions that capture the logistic requirements in real-world.

\bibliography{refs}

\end{document}